\documentclass[12pt,english]{article}
\usepackage[T1]{fontenc}
\usepackage[latin9]{inputenc}
\usepackage{geometry}
\usepackage{color}
\usepackage{float}
\usepackage{amsmath}
\usepackage{amssymb}
\usepackage{graphicx}
\usepackage{setspace}
\usepackage[authoryear]{natbib}
\makeatletter

\providecommand{\tabularnewline}{\\}

\date{}

\bibpunct{(}{)}{;}{autheryear}{,}{;}

\newtheorem{assumption}{Assumption}
\newtheorem{step}{Step}
\newtheorem{theorem}{Theorem}
\newtheorem{remark}{Remark}

\newtheorem{proposition}{Proposition}
\newtheorem{definition}{Definition}
\newtheorem{example}{Example}
\newtheorem{lemma}{Lemma}

\newcommand*{\indep}{%
  \mathbin{%
    \mathpalette{\@indep}{}%
  }%
}
\newcommand*{\nindep}{%
  \mathbin{%                   % The final symbol is a binary math operator
    \mathpalette{\@indep}{\not}% \mathpalette helps for the adaptation
  }%
}
\newcommand*{\@indep}[2]{%
  \sbox0{$#1\perp\m@th$}%        box 0 contains \perp symbol
  \sbox2{$#1=$}%                 box 2 for the height of =
  \sbox4{$#1\vcenter{}$}%        box 4 for the height of the math axis
  \rlap{\copy0}%                 first \perp
  \dimen@=\dimexpr\ht2-\ht4-.2pt\relax
  \kern\dimen@
  {#2}%
  \kern\dimen@
  \copy0 %                       second \perp
} 

\newcommand{\T}{\mathrm{\scriptscriptstyle T}}
\newcommand{\pr}{\mathrm{pr}} 
\newcommand{\obs}{{r}}
\newcommand{\mis}{{\overline r}}
\newcommand{\obsR}{{R}}
\newcommand{\misR}{{\overline R}}
\newcommand{\R}{\mathcal{R}}
\newcommand{\de}{\mathrm{d}}

\makeatother

\usepackage{babel}
\begin{document}
\title{Causal inference with confounders missing not at random}
\author{Shu Yang\thanks{Department of Statistics, North Carolina State University, North Carolina
27695, U.S.A.}, Linbo Wang\thanks{Department of Statistical Sciences, University of Toronto, Toronto, Ontario M5S 3G3, Canada}, Peng Ding\thanks{Department of Statistics, University of California, Berkeley, California
94720, U.S.A.}}
\maketitle
\begin{abstract}
It is important to draw causal inference from 
observational studies, which, however, becomes challenging if the
confounders have missing values. Generally, causal effects
are not identifiable if the confounders are missing not at random.
We propose a novel framework to nonparametrically identify causal effects with confounders subject to an outcome-independent
missingness, that is, the missing data mechanism is independent of
the outcome, given the treatment and possibly missing confounders.
We then propose a nonparametric two-stage least squares estimator and a
parametric estimator for causal effects.

\textit{Keywords}:  Completeness; Ill-posed inverse problem;
Integral equation;  Outcome-independent missingness 
\end{abstract}

\section{Introduction }

Causal inference plays important roles
in biomedical studies and social sciences. 
%If we observe all the confounders
%of the treatment-outcome relationship, we can use standard techniques,
If all the confounders
of the treatment-outcome relationship are observed, one can use standard techniques,
such as propensity score matching, subclassification and weighting
to adjust for confounding \citep[e.g.,][]{rosenbaum1983central,imbens2015causal}.

Much less work has been done to deal with the case when confounders have missing values.
\citet{rosenbaum1984reducing}
and \citet{d2000estimating} developed a generalized propensity score
approach. Under a modified unconfoundedness assumption, they showed
that adjusting for the missing pattern and the observed values of
confounders removes all confounding bias, and hence the causal effects are
identifiable. Moreover, the balancing property of the propensity score
carries over to the generalized propensity
score. Standard propensity score methods can hence be used to estimate the
causal effects. However, the modified unconfoundedness assumption implies
that units may have different confounders depending on the missing
pattern, which is often hard to justify scientifically. An alternative
approach assumes that the confounders are missing at random \citep{rubin1976inference}.
Under this assumption, both the full data distribution and causal
effects are identifiable, and multiple imputation can provide reasonable estimates of the causal effects \citep{rubin1987multiple,qu2009propensity,crowe2010comparison,mitra2011estimating,seaman2014inverse}.
In practice, however, the missing pattern often depends on the missing
values themselves, a scenario commonly known as missing not at random
\citep{rubin1976inference}. Previous multiple imputation methods may fail to provide
valid inference in this scenario. See \citet{mattei2009estimating}
for a comparison of various methods and \citet{lu2017propensity}
for a sensitivity analysis.

Causal inference with confounders missing not at random is challenging
because neither the full data distribution nor the causal effects
are identifiable without further assumptions. We consider
a novel setting where the confounders are subject to an outcome-independent
missingness, that is, the missing data mechanism is independent of
the outcome, given the treatment and possibly missing confounders.
This outcome-independent missingness is more plausible if the outcome
happens after the covariate measurements and missing data indicators. To identify
the causal effects under this setting, we formulate the identification
problem as solving an integral equation, and show that identification
of the full data distribution is equivalent to unique existence
of the solution to an inverse problem. This new perspective allows
us to establish a general condition for identifiability of the causal
effects. Our condition generalizes the existing results for a discrete covariate and outcome \citep{ding2014identifiability}. Motivated by
the identification result, we develop a nonparametric two-stage least
squares estimator by solving the sample analog of the integral equation.
To avoid the curse of dimensionality, we further develop parametric likelihood-based methods.

\section{Setup and assumptions\label{sec:Setup}}

\subsection{Potential outcomes, causal effects, and unconfoundedness\label{subsec:POs} }

We use potential outcomes to define causal effects. Suppose that the
binary treatment is $A\in\{0,1\}$, with $0$ and $1$ being the labels
for the control and active treatments, respectively. Each level
of treatment $a$ corresponds to a potential outcome $Y(a)$, representing
the outcome had the subject, possibly contrary to the fact, been given
treatment $a$. The observed outcome is $Y=Y(A)=AY(1)+(1-A)Y(0)$.
Let $X=(X_{1},\ldots,X_{p})$ be a vector of $p$-dimensional pre-treatment
covariates. We assume that a sample of size $n$ consists of independent
and identically distributed draws from the distribution of $\{A,X,Y(0),Y(1)\}$.
The covariate-specific causal effect is $\tau(X)=E\{Y(1)-Y(0)\mid X\}$,
and the average causal effect is $\tau=E\{Y(1)-Y(0)\}=E\{\tau(X)\}$.
We focus on $\tau$, and a similar discussion applies to the average
causal effect on the treated $\tau_{\mathrm{ATT}}=E\{Y(1)-Y(0)\mid A=1\}=E\{\tau(X)\mid A=1\}$.
%These causal effects are not identifiable without further assumptions,
%because for each subject, only one potential outcome is observed.
The following assumptions are standard in causal inference with observational
studies \citep{rosenbaum1983central}.

\begin{assumption} \label{asump-ignorable} $\{Y(0),Y(1)\}\indep A\mid X$.
\end{assumption}

\begin{assumption} \label{asump-overlap}There exist constants $c_{1}$
	and $c_{2}$ such that $0<c_{1}\leq e(X)\leq c_{2}<1$ almost surely,
	where $e(X)=\pr(A=1\mid X)$ is the propensity score.
	
\end{assumption}

Under Assumptions \ref{asump-ignorable} and \ref{asump-overlap},
$\tau=E\{E(Y\mid A=1,X)-E(Y\mid A=0,X)\}$ is identifiable from the joint
distribution of observed data $(A,X,Y)$. \citet{rosenbaum1983central}
showed that $\{Y(0),Y(1)\}\indep A\mid e(X)$, so adjusting for the
propensity score removes all confounding. We can estimate $\tau$ through propensity score matching, subclassification
or weighting. 

\subsection{Confounders with missing values and the generalized propensity score}

We consider the case where $X$ contains missing values. Let $R=(R_{1},\ldots,R_{p})$
be the vector of missing indicators such that $R_{j}=1$ if the $j$th
component $X_{j}$ is observed and $0$ if $X_j$ is missing. Let $\mathcal{R}$
be the set of all possible values of $R$. We use $1_{p}$ to denote
the $p$-vector of $1$'s and $0_{p}$ to denote the $p$-vector of
$0$'s. The missingness pattern $R=r \in \mathcal{R}$ partitions the covariates $X$ into $X_{r}$ and $X_{\overline{r}}$, the observed and missing parts of $X$, respectively. Using \citet{rubin1976inference}'s notation, $X_{\obsR} = X_{\text{obs}}$ and $ X_{\misR} = X_{\text{mis}} $ are the realized observed and missing covariates, respectively. 
For example, if $R_{1}=1$ and $R_{j}=0$ for $j=2,\ldots,p$, then
$X_{\obsR}=X_{1}$ and $X_{\misR}=(X_{2},\ldots,X_{p})$. With missing
data, assume that we have independent and identically distributed
draws from $\{A,X,Y(1),Y(0),R\}$. \citet{rosenbaum1984reducing}
introduced the following modified unconfoundedness assumption.

\begin{assumption} \label{asump:M1-1}$\{Y(0),Y(1)\}\indep A\mid(X_{\obsR},R)$.
\end{assumption}

Under Assumption \ref{asump:M1-1}, the generalized propensity score
$e(X_{\obsR},R) =\pr(A=1\mid X_{\obsR},R)$
plays the same role as the usual propensity score $e(X) = \pr(A=1\mid X)$ in the settings
without missing covariates. \citet{rosenbaum1984reducing} showed
that adjusting for $e(X_{\obsR},R)$ balances $(X_{\obsR},R)$
and removes all confounding on average. Their approach has
the advantage of requiring no assumptions on the missing data mechanism of
$X$ for the identification of causal effects. However, their
approach implies that a pre-treatment covariate can be a confounder when
it is observed but not a confounder when it is missing. This is
often hard to justify scientifically. Moreover, if the covariate measurement
occurs after the treatment assignment, then $R$ is a post-treatment
variable affected by $A$. In this case, even if $A$ is completely randomized, Assumption \ref{asump:M1-1}
is unlikely to hold conditioning on the post-treatment variable $R$ \citep{frangakis2002principal}.

\subsection{Missing data mechanisms of the confounders \label{subsec:Missing-data-mechanisms}}

%We now describe existing estimation methods with explicit assumptions
%on the missing data mechanism of $X$. 

%Without Assumption \ref{asump:M1-1}, we need to impose assumptions on the missing data mechanism.
Without Assumption \ref{asump:M1-1}, one needs to impose assumptions on the missing data mechanism.
We now describe existing estimation methods under different missingness mechanisms of the confounders.
The first one is missing completely
at random \citep{rubin1976inference}.

\begin{assumption}[Missing completely at random]\label{asump:MCAR}$R\indep(A,X,Y)$.
\end{assumption}

Assumption \ref{asump:MCAR} requires that the missingness of confounders is independent of all variables $(A,X,Y)$.
It implies $\tau=E\{\tau(X)\mid R=1_{p}\}$
and thus justifies the complete-case analysis that uses only the units
with fully observed confounders. This complete-case analysis is however
inefficient by discarding all units with missing confounders. Moreover,
confounders are rarely missing completely at random.

The second one is missing at random \citep{rubin1976inference}.

\begin{assumption}[Missing at random]\label{asump:MAR}$R\indep X\mid(A,Y)$.
\end{assumption}

Under Assumption \ref{asump:MAR}, conditioning on the treatment and
outcome, the missing mechanism of confounders is independent of the
missing values themselves. Assumption \ref{asump:MAR} implies $f(A,X,Y)=f(A,Y)f(X\mid A,Y,R=1_{p})$,
and therefore, the joint distribution and its functionals including
$\tau$ are all identifiable. \citet{rubin1976inference} showed that
we can ignore the missing data mechanism in the likelihood-based and Bayesian
inferences under Assumption \ref{asump:MAR}. In this case, multiple
imputation is a popular tool for causal inference \citep[e.g.,][]{qu2009propensity,crowe2010comparison,mitra2011estimating,seaman2014inverse}. 
%In general, imputing the missing confounders based on $f(X_\mis  \mid X_\obs, A, Y) \propto f(X)f(A\mid X)f(Y\mid A,X)$ involves an outcome model, which violates the principle of not using the outcome in the design of an observational study \citep[cf.][]{rubin2007design, clearinghouse2017works}. 
Although \cite{rubin2007design} suggested not using the outcome in the design of an observational study, imputing the missing confounders based on $f(X_\misR  \mid X_\obsR, A, Y) \propto f(X)f(A\mid X)f(Y\mid A,X)$ involves an outcome model in general \citep{clearinghouse2017works}.

However, missingness at random is not plausible if the missing pattern
depends on the missing values themselves. Instead, we consider the following
missing data mechanism.

\begin{assumption}[Outcome-independent missingness]\label{asump:OIMissing}$R\indep Y\mid(A,X)$.
	
\end{assumption}

Assumption \ref{asump:OIMissing} is plausible for prospective
observational studies with $X$ measured long before the outcome takes
place.
%; see $\mathsection$6 for motivating examples. 
Moreover, Assumption \ref{asump:OIMissing} is more plausible than Assumptions 4 and 5 for a certain class of examples where a potentially hazardous exposure has come under substantial scrutiny, data may be collected more comprehensively for exposed than for unexposed subjects; e.g., for the water crisis in Flint, Michigan U.S. \citep{hanna2016elevated}, potentially exposed children and neighborhoods will have been more carefully measured than unexposed children and neighborhoods that will eventually serve as their comparisons.

Figure \ref{fig:1} is a causal diagram \citep{pearl1995causal}
illustrating Assumptions \ref{asump-ignorable} and \ref{asump:OIMissing}.
Graphically, $A$ and $Y$ have no common parents except for $X$,
encoding Assumption \ref{asump-ignorable}, and $R$ and $Y$ have
no common parents except $A$ and $X$, encoding Assumption \ref{asump:OIMissing}.
Our framework allows for unmeasured common causes of $R$ and $A$,
and the dependence of $R$ on the missing confounders $X_{\misR}$. Moreover, it allows $R$
to be a post-treatment variable affected by $A$. 

Assumption \ref{asump:OIMissing} exploits the temporal orders of $(A,X,R)$ and $Y$. It restricts the joint distribution of $(A,X,R,Y)$ which makes nonparametric identification possible. This is a feature of the missingness of confounders. In contrast, the missingness of outcome may depend on all variables happening before. 
%We do not consider the case with outcomes missing not at random because it requires stronger assumptions. 

%As a motivating
%example, we estimate the causal effect of education on general health
%satisfaction using the dataset from the 2015\textendash 2016 U.S.
%National Health and Nutrition Examination Survey. Important confounders,
%such as the family poverty ratio, have missing values. It is implausible
%to have missingness at random because subjects with low family poverty
%ratios are less likely to disclose their status. Moreover, the missingness
%is not perceivably related to the general health satisfaction given
%the education level and the confounders.

\begin{figure}
	% The arguments in the next line are {height}{optional width}{used only by OUP for typesetting}[filename, in directory art]
\begin{centering}
	\includegraphics[scale=0.5]{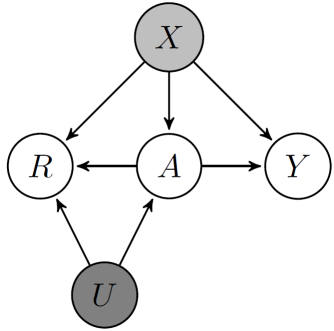}
	\par\end{centering}
	%\figurebox{9pc}{}{}[CI6.eps]
	% note that files may not be rotated
	\caption{A direct acyclic graph illustrating Assumptions \ref{asump-ignorable} and \ref{asump:OIMissing}.
		White nodes represent observed variables, the light grey node represents
		the variable with missing values, and the dark node represents an unmeasured
		variable $U$.}
	\label{fig:1}
\end{figure}

We also make the following assumption to rule out degeneracy of the
missing data mechanism.

\begin{assumption}\label{asump:Pos} $\pr(R=1_{p}\mid A,X,Y)>c_{3}>0$
	almost surely for some constant $c_{3}.$ \end{assumption}

\section{Nonparametric identification\label{sec:Nonparametric-identification} }

\subsection{Identification strategy}
Assume that the distribution of $(A,X,Y,R)$ is absolutely continuous
with respect to some measure, with $f(A,X,Y,R)$ being the density or
probability mass function. Under Assumptions \ref{asump-ignorable}
and \ref{asump-overlap}, the key is to identify the joint distribution
of $f(A,X,Y)$ because $\tau$ is its functional. The following identity
relates the full data distribution to the observed data distribution:
\begin{equation}
	f(A,X,Y,R=1_{p})=f(A,X,Y)\pr(R=1_{p}\mid A,X,Y).\label{eq:selm}
\end{equation}
The left-hand side of (\ref{eq:selm}) is identifiable under Assumption
\ref{asump:Pos}. Therefore, the identification of $f(A,X,Y)$
relies on the identification of $\pr(R=1_{p}\mid A,X,Y)$. 
%For example,
%under Assumption \ref{asump:MAR}, $\pr(R=1_{p}\mid A,X,Y)=\pr(R=1_{p}\mid A,Y)$
%is identifiable. 
We now discuss how to identify $\pr(R=1_{p}\mid A,X,Y)=\pr(R=1_{p}\mid A,X)$
under Assumption \ref{asump:OIMissing}.

\subsection{Integral equation representation }

Under Assumption \ref{asump:OIMissing}, let 
\[
\xi_{ra}(X)=\frac{\pr(R=r\mid A=a,X,Y)}{\pr(R=1_{p}\mid A=a,X,Y)}=\frac{\pr(R=r\mid A=a,X)}{\pr(R=1_{p}\mid A=a,X)},\quad(a=0,1;r\in\mathcal{R}).
\]
It then suffices to identify $\xi_{ra}(X)$, because it determines
the missing data mechanism via 
\begin{equation}
	\pr(R=r\mid A=a,X,Y)=\frac{\pr(R=r\mid A=a,X,Y)}{\sum_{r'\in\mathcal{R}}\pr(R=r'\mid A=a,X,Y)}=\frac{\xi_{ra}(X)}{\sum_{r'\in\mathcal{R}}\xi_{r'a}(X)}.\label{eq:missinf prob}
\end{equation}
The following theorem shows that $\xi_{ra}(X)$ is a key term connecting
the observed data distribution $f(A,X_{\obs},Y,R)$ and the complete-case
distribution $f(A,X,Y,R=1_{p})$.
Throughout the paper, we use $\nu(\cdot)$ to denote a generic measure, e.g., the Lebesgue measure for continuous variable
and the counting measure for discrete variable.

\begin{theorem}\label{thm::basis} Under Assumption \ref{asump:OIMissing},
	for any $r$ and $a$, the following integral equation holds: 
	\begin{equation}
		f(A=a,X_{\obs},Y,R=r)=\int\xi_{ra}(X)f(A=a,X,Y,R=1_{p})\de\nu(X_{\mis}).\label{eq:m2}
	\end{equation}
\end{theorem}

\begin{proof} The conclusion follows because observed data distribution
	is the complete data distribution averaged over the missing data:
	\begin{eqnarray*}
		f(A=a,X_{\obs},Y,R=r) & = & \int f(A=a,X,Y,R=r)\de\nu(X_{\mis})\\
		& = & \int\frac{\pr(R=r\mid A=a,X,Y)}{\pr(R=1_{p}\mid A=a,X,Y)}f(A=a,X,Y,R=1_{p})\de\nu(X_{\mis})\\
		& = & \int\xi_{ra}(X)f(A=a,X,Y,R=1_{p})\de\nu(X_{\mis}).
	\end{eqnarray*}
\end{proof}

%\begin{remark} In \eqref{eq:m2} and the proof above, $\nu(\cdot)$
%	is a generic measure, e.g., the Lebesgue measure for continuous variable
%	and the counting measure for discrete variable. For example, if $X=(X_{1},X_{2})$
%	is continuous and $r=(1,0)$ determines that $X_{\obs}=X_{1}$
%	and $X_{\mis}=X_{2}$, then \eqref{eq:m2} means $f(A=a,X_{1},Y,R=r_{1})=\int\xi_{ra}(X_1,X_2)f(A=a,X_1,X_2,Y,R=1_{2})\de\nu(X_{2})$.
%\end{remark}

Theorem \ref{thm::basis} is the basis of our identification analysis.
In (\ref{eq:m2}), $f(A=a,X_{\obs},Y,R=r)$ and $f(A=a,X,Y,R=1_{p})$
are identifiable from the observed data. We have thus turned the identification
of $\xi_{ra}(X)$ to the problem of solving $\xi_{ra}(X)$ from (\ref{eq:m2}).
This requires additional technical assumptions below.

\subsection{Bounded completeness and identification of the joint distribution }

To motivate our identification conditions, we first consider the case
with discrete $X$ and $Y$, where (\ref{eq:m2}) becomes a linear
system. To solve $\xi_{ra}(X)$ from (\ref{eq:m2}), we need the linear
system to be non-degenerate.

\begin{proposition}\label{proposition1}Under Assumption \ref{asump:OIMissing},
	suppose that $X$ and $Y$ are discrete, and $X_{j}\in\{x_{j1},\ldots x_{jJ_{j}}\}$
	for $j=1,\ldots,p$, and $Y\in\{y_{1},\ldots,y_{K}\}$. Let $q=J_{1}\times\cdots\times J_{p}$,
	and let $\Theta_{a}$ be a $K\times q$ matrix with the $k$-th row
	being $f(X,y_{k},R=1_{p},A=a)$ evaluated at all possible values of
	$X$. The distribution of $(A,X,Y,R)$ is identifiable if $\textup{Rank}(\Theta_{a})=q$
	for $a=0,1$. \end{proposition}

We relegate the proof to the Supplementary Material. For the special case with a binary $X$ and a discrete $Y$, the rank condition in Proposition \ref{proposition1} is equivalent to $X \nindep Y\mid  (A=a, R=1)$ for $a=0$ and $1$, which is testable based on the observed data \citep{ding2014identifiability}. For general cases, we need to extend the rank condition that ensures the unique existence of $\xi_{ra}(X)$. We use the notion of bounded completeness for general $X$ and $Y$, which is related to the concept of a complete statistic \citep{lehmann1950completeness, newey2003instrumental}. Below, we say that a function $g(x)$ is bounded in $\mathcal{L}_{1}$-metric
if $\sup_{x}|g(x)|\leq c$ for some $0<c<\infty$.

\begin{definition}\label{def:bddcomp}A function
	$f(X,Y)$ is bounded complete in $Y$, if $\int g(X)f(X,Y)\de\nu(X)=0$
	implies $g(X)=0$ almost surely for any measurable function $g(X)$
	bounded in $\mathcal{L}_{1}$-metric. \end{definition}

\citet{d2011completeness} gave sufficient conditions for the bounded
completeness. It also appeared in other identification analyses such
nonparametric instrumental variable regression models \citep{darolles2011nonparametric}
and measurement error models \citep{an2012well}.

We invoke the following assumption motivated by Theorem \ref{thm::basis}
and Definition \ref{def:bddcomp}.

\begin{assumption}\label{asump:bddcomp} The joint distribution $f(A=a,X,Y,R=1_{p})$
	is bounded complete in $Y$, for $a=0,1$.
	
\end{assumption}

\begin{remark}\label{rmk: bddcomp}When $X$ and $Y$ are discrete
	with finite supports, Assumption \ref{asump:bddcomp} is equivalent
	to the rank condition in Proposition \ref{proposition1}. Moreover,
	Assumption \ref{asump:bddcomp} implies Assumption \ref{asump-overlap}.
\end{remark}

Under Assumption \ref{asump:Pos}, Assumption \ref{asump:bddcomp}
is  sufficient to ensure the unique existence of
$\xi_{ra}(X)$ from (\ref{eq:m2}). We present the result in the following
theorem.

\begin{theorem}\label{thm2}Under Assumptions \ref{asump:OIMissing}\textendash \ref{asump:bddcomp},
	the distribution of $(A,X,Y,R)$ is identifiable. \end{theorem}

\begin{proof} Suppose that $\xi_{ra}^{(1)}(X)$ and $\xi_{ra}^{(2)}(X)$
	are two solutions to \eqref{eq:m2}: 
	\[
	f(A=a,X_{\obs},Y,R=r)=\int\xi_{ra}^{(k)}(X)f(A=a,X,Y,R=1_{p})\de\nu(X_{\mis}),\quad(k=1,2)
	\]
	which imply $\int\{\xi_{ra}^{(1)}(X)-\xi_{ra}^{(2)}(X)\}f(A=a,X,Y,R=1_{p})\de\nu(X_{\mis})=0.$
	Integrating this identity with respect to $X_{\obs}$, we have 
	\[
	\int\{\xi_{ra}^{(1)}(X)-\xi_{ra}^{(2)}(X)\}f(A=a,X,Y,R=1_{p})\de\nu(X)=0.
	\]
	Assumption \ref{asump:Pos} implies that $\xi_{ra}(X)$ is bounded
	in $\mathcal{L}_{1}$-metric, which further implies that $\xi_{ra}^{(1)}(X)-\xi_{ra}^{(2)}(X)$
	is bounded in $\mathcal{L}_{1}$-metric. Under Assumption \ref{asump:bddcomp},
	Definition \ref{def:bddcomp} implies that $\xi_{ra}^{(1)}(X)-\xi_{ra}^{(2)}(X)=0$
	almost surely. Therefore, \eqref{eq:m2} has a unique solution $\xi_{ra}(X)$.
	Based on the definition of $\xi_{ra}(X)$, we can identify $\pr(R=1_{p}\mid A,X,Y)$
	by \eqref{eq:missinf prob}. Finally, we identify $f(A,X,Y)$ through
	\eqref{eq:selm} as $f(A,X,Y)=f(A,X,Y,R=1_{p})/\pr(R=1_{p}\mid A,X,Y).$
\end{proof}

If the distribution of $(A,X,Y)$ is identifiable, we can use a
standard argument to show that $\tau$ and $\tau_{\mathrm{ATT}}$ are identifiable
under Assumption \ref{asump-ignorable}.
We give explicit identification formulas for $\tau$ and $\tau_{\mathrm{ATT}}$ in the next subsection, 
which are the basis for constructing the nonparametric estimator.

\subsection{Nonparametric identification formulas for average causal effects}

Under Assumptions \ref{asump-ignorable},
\ref{asump:OIMissing}\textendash \ref{asump:bddcomp}, we can identify
$\tau$ and $\tau_{\mathrm{ATT}}$ in two steps. First, 
\begin{eqnarray}
	\tau(X) & = & E(Y\mid A=1,X)-E(Y\mid A=0,X)\label{eq::tauX1}\\
	& = & E(Y\mid A=1,X,R=1_{p})-E(Y\mid A=0,X,R=1_{p}),\label{eq:tau(X)}
\end{eqnarray}
where \eqref{eq::tauX1} follows from Assumption \ref{asump-ignorable},
and \eqref{eq:tau(X)} follows from Assumption \ref{asump:OIMissing}.
Therefore, we can identify $\tau(X)$ using a complete-case analysis
based on \eqref{eq:tau(X)}.

Second, under Assumptions \ref{asump:OIMissing}\textendash \ref{asump:bddcomp},
Theorem \ref{thm2} shows that the distribution of $(A,X,Y,R)$ is
identifiable, which implies that the marginal distribution of $X$,
$f(X)$, and the conditional distribution of $X$, $f(X\mid A=1)$,
are also identifiable. Therefore, both $\tau=E\{\tau(X)\}$ and $\tau_{\mathrm{ATT}}=E\{\tau(X)\mid A=1\}$
are identifiable. The following theorem summarizes these results.

\begin{theorem}\label{Thm:tau}Under Assumptions \ref{asump-ignorable},
	\ref{asump:OIMissing}\textendash \ref{asump:bddcomp},
	the average causal effect $\tau$ is identified by 
	\begin{eqnarray}
		\tau=\sum_{a=0}^{1}\int\tau(X)\frac{f(A=a,X,R=1_{p})}{\pr(R=1_{p}\mid A=a,X)}\de\nu(X),\label{eq::ate}
	\end{eqnarray}
	and the average treatment effect on the treated $\tau_{\mathrm{ATT}}$
	is identified by 
	\begin{eqnarray}
		\tau_{\mathrm{ATT}}=\int\tau(X)\frac{f(X,R=1_{p}\mid A=1)}{\pr(R=1_{p}\mid A=1,X)}\de\nu(X).\label{eq::att}
	\end{eqnarray}
	where $\tau(X)$ is identified by (\ref{eq:tau(X)}), $\pr(A=a,R=1_{p})$
	and $f(A=a,X,R=1_{p})$ depend only on the observed data, and $\pr(R=1_{p}\mid A=a,X)$
	can be identified from (\ref{eq:m2}), for $a=0$ and $1$. \end{theorem}

\begin{proof} First, we can identify the conditional distribution
	of $X$ given $A=a$ by 
	\[
	f(X\mid A=a)=\frac{f(X,R=1_{p}\mid A=a)}{\pr(R=1_{p}\mid A=a,X)},\quad(a=0,1).
	\]
	Averaging $\tau(X)$ over $f(X\mid A=1)$, we obtain the identification
	formula \eqref{eq::att}.
	
	Second, we can identify the marginal distribution of $X$ by 
	\[
	f(X)=\sum_{a=0}^{1}f(A=a,X)=\sum_{a=0}^{1}\frac{f(A=a,X,R=1_{p})}{\pr(R=1_{p}\mid A=a,X)}.
	\]
	Averaging $\tau(X)$ over the above distribution, we obtain the identification
	formula \eqref{eq::ate}. \end{proof}

\section{Estimation of the average causal effect\label{sec::estimation} }

\subsection{Nonparametric two-stage least squares estimator\label{sec:Nonparametric-estimation}}

Theorem \ref{Thm:tau} shows the nonparametric identification formulas
at the population level. Based on \eqref{eq::ate}, we propose a nonparametric
two-stage least squares estimator of $\tau$ with finite samples $(A_{i},X_{i},Y_{i},R_{i})_{i=1}^{n}$. We omit the estimation of $\tau_{\mathrm{ATT}}$. We estimate $\tau(X)$, $\pr(A=a,R=1_{p})$, $f(X\mid A=a,R=1_{p})$
and $\pr(R=1_{p}\mid A=a,X)$ by standard nonparametric methods, denoted
by $\hat{\tau}(X)$, $\widehat{\pr}(A=a,R=1_{p})$, and $\hat{f}(X\mid A=a,R=1_{p})$,
respectively. Thus, the key is to estimate $\pr(R=1_{p}\mid A=a,X)$,
or, equivalently, $\xi_{ra}(X)$ based on (\ref{eq:m2}).

In the first stage, we obtain $\hat{f}(X_{\obs},Y,R=r\mid A=a)$ and
$\hat{f}(X,Y,R=1_{p}\mid A=a)$ as the nonparametric sample analogs
of $f(X_{\obs},Y,R=r\mid A=a)$ and $f(X,Y,R=1_{p}\mid A=a)$. Replacing
them in (\ref{eq:m2}) leads to 
\begin{equation}
	\hat{f}(X_{\obs},Y,R=r\mid A=a)=\int\xi_{ra}(X)\hat{f}(X,Y,R=1_{p}\mid A=a)\de\nu(X_{\mis}),\label{eq:m3}
\end{equation}
which is a Fredholm integral equation of the first kind. Solving (\ref{eq:m3})
raises several challenges. First, although Theorem \ref{thm2} shows
that the population equation (\ref{eq:m2}) has a unique solution,
the sample equation (\ref{eq:m3}) may not. Second, $\xi_{ra}(X)$
is an infinite-dimensional parameter, and its estimation often relies
on some approximation. Third, solving $\xi_{ra}(X)$ from (\ref{eq:m3})
is an ill-conditioned problem, in the sense that even a slight perturbation
of $\hat{f}(X_{\obs},Y,R=r\mid A=a)$ and $\hat{f}(X,Y,R=1_{p}\mid A=a)$
can lead to a large variation in the solution for $\xi_{ra}(X)$.
As a result, replacing $f(X_{\obs},Y,R=r\mid A=a)$
and $f(X,Y,R=1_{p}\mid A=a)$ \eqref{eq:m2} by their consistent estimators does not necessarily
yield a consistent estimator of $\xi_{ra}(X)$ \citep{darolles2011nonparametric}.

To tackle these issues, we use the series approximation \citep{kress1989linear,newey2003instrumental}
in the second stage. Let the set $\mathcal{H}_{J}=\{h^{j}(X)=\exp(-X^{\T}X)X^{\lambda_{j}}:j=1,\ldots,J\}$
form a Hermite polynomial basis, where $X^{\lambda_{j}}=X_{1}^{\lambda_{j1}}\cdots X_{p}^{\lambda_{jp}}$,
with $\lambda_{j}=(\lambda_{j1},\ldots,\lambda_{jp})$ and $|\lambda_{j}|=\sum_{l=1}^{p}\lambda_{jl}$
increasing in $j$. Let $\tilde{X}=\Sigma^{-1/2}(X-\mu)$ be a standardized
version of $X$, where $\mu$ and $\Sigma$ are constant vector and
matrix. We approximate $\xi_{ra}(X)$ by $\xi_{ra}(X)\approx\sum_{j=1}^{J}\beta_{ra}^{j}h^{j}(\tilde{X}).$
Thus, for each missing pattern $R=r$, we approximate \eqref{eq:m2}
by 
\begin{eqnarray}
	f(X_{\obs},Y,R=r\mid A=a) & \approx & \sum_{j=1}^{J}\beta_{ra}^{j}\int h^{j}(\tilde{X})f(X,Y,R=1_{p}\mid A=a)\de\nu(X_{\mis})\nonumber \\
	& = & \sum_{j=1}^{J}\beta_{ra}^{j}H_{ra}^{j}(X_{\obs},Y)f(X_{\obs},Y,R=1_{p}\mid A=a),\label{eq::sample-approx}
\end{eqnarray}
where the conditional expectation $H_{ra}^{j}(X_{\obs},Y)=E\{h^{j}(\tilde{X})\mid A=a,X_{\obs},Y,R=1_{p}\}$
is over the distribution $f(X_{\mis}\mid A=a,X_{\obs},Y,R=1_{p})$.

We need the empirical versions of $H_{ra}^{j}(X_{\obs},Y)$ and $f(X_{\obs},Y,R=1_{p}\mid A=a)$
for estimation. First, for unit $i$, let $\hat{H}_{ra,i}^{j}=\hat{E}\{h^{j}(\tilde{X})\mid A_{i}=a,X_{\obs,i},Y_{i},R_{i}=1_{p}\}$
be a nonparametric estimator of the conditional expectation. 
Second, we obtain $\hat{f}(X_{\obs},Y,R=1_{p}\mid A=a)$,
a nonparametric estimator of $f(X_{\obs},Y,R=1_{p}\mid A=a)$, based
on the complete cases. 
Although we obtain these estimators based on the complete cases, we still
need to partition the confounders into $(X_{\obs},X_{\mis})$ based
on the missing pattern $R=r$. 
Because the sample version of the approximation
\eqref{eq::sample-approx} is linear, we can estimate the $\beta_{ra}^{j}$'s
by minimizing the residual sum of squares 
\begin{equation}
	\sum_{i=1}^{n}\left\{ \hat{f}(X_{\obs,i},Y_{i},R_{i}=r\mid A_{i}=a)\vphantom{\sum_{R=1}^{\infty}}-\sum_{j=1}^{J}\beta_{ra}^{j}\hat{H}_{ra,i}^{j}\hat{f}(X_{\obs,i},Y_{i},R_{i}=1_{p}\mid A_{i}=a)\right\} ^{2}.\label{eq:obj fctn}
\end{equation}

To solve the ill-conditioned problem, we restrict the parameter space
of $\xi_{ra}(X)$ to a compact space, which can effectively regularizes
the problem to be well posed. 
%This is because integration is a continuous
%operator, and restricting to a compact space makes its inverse be
%continuous. 
Given the approximation of $\xi_{ra}(X)$, we require
the vector of coefficients $\beta_{ra}$, the concatenation of $(\beta_{ra}^{1},\ldots,\beta_{ra}^{J})$,
satisfy $\beta_{ra}^{\T}\Lambda\beta_{ra}\leq B$,
where $\Lambda$ is a positive definite $J\times J$ matrix and $B$
is a positive constant. Therefore, we propose to estimate $\beta_{ra}$
by minimizing (\ref{eq:obj fctn}), subject to the constraint $\beta_{ra}^{\T}\Lambda\beta_{ra}\leq B$.
We present
more details of regularization in the Supplementary Material.

We then estimate $\xi_{ra}(X)$ and the probability $\pr(R=1_{p}\mid A=a,X)$
by 
\[
\hat{\xi}_{ra}(X)=\sum_{j=1}^{J}\hat{\beta}_{ra}^{j}h^{j}(\tilde{X}),\quad\widehat{\pr}(R=1_{p}\mid A=a,X)=\left\{ 1+\sum_{r\neq1_{p}}\hat{\xi}_{ra}(X)\right\} ^{-1},
\]
and finally estimate $\hat{\tau}$ by 
\begin{equation}
	\sum_{a=0}^{1}\widehat{\pr}(A=a,R=1_{p})\int\hat{\tau}(X)\frac{\hat{f}(X\mid A=a,R=1_{p})}{\widehat{\pr}(R=1_{p}\mid A=a,X)}\de\nu(X).\label{eq:tauhat}
\end{equation}

We now comment on subtle technical issues for implementing the above estimator.
First, we need to standardize the confounders by $\tilde{X}=\Sigma^{-1}(X-\mu)$.
We choose $\mu$ and $\Sigma$ to be the mean and covariance matrix
of confounders for the complete cases. This choice is innocuous because
$\mathcal{H}_{J}$ remains the same for other values of $\mu$ and $\Sigma$. Second, we
use the importance sampling technique to approximate the integral
in \eqref{eq:tauhat} because it is difficult to directly sample from
the nonparametric density estimators. Third, we use the bootstrap
to construct confidence intervals. 
%%%%%%%%%%%%%%%%%%%%
\cite{newey1997convergence} provided a relatively simple variance estimation approach treating
the nonparametric estimators as if they were parametric given the
fixed tuning parameters. 
In the light of treating the
nonparametric estimators as if they were parametric, one might expect
the nonparametric bootstrap to work for our estimator; see, e.g., \cite{horowitz2007asymptotic}. For all bootstrap
samples, we use the same tuning parameters, such as the smoothing
parameter in the smoothing splines and the bandwidth in the kernel
density estimator. 
%%%%%%%%%%
%We use
%the same tuning parameters, such as the smoothing parameter in the
%smoothing spline estimator and the bandwidth in the kernel density
%estimator, for all bootstrap samples. 
In the Supplementary Material, we give more technical details and explicate the procedure in an example with a scalar confounder. 

%Under regularity conditions, the proposed estimator is consistent. The simulation in \S \ref{subsec:Confounder}
%will show that although the proposed estimator has small finite-sample biases, other existing
%estimators including the generalized propensity score weighting estimator
%and multiple imputation estimators have larger biases when the confounder
%is subject to the outcome-independent missingness.

\subsection{Parametric estimation: likelihood-based and Bayesian inferences\label{sec:Parametric-estimation}}

The nonparametric estimator above suffers from the curse of dimensionality.
We propose a parametric estimation for moderate- or high-dimensional covariates. Let $Z_{i}=(A_{i},X_{i},Y_{i},R_{i})$
be the complete data and $Z_{\obsR,i}=(A_{i},X_{\obsR,i},Y_{i},R_{i})$
be the observed data for unit $i $. The complete-data
likelihood is $L(\theta\mid Z_{1},\ldots,Z_{n})=\prod_{i=1}^{n}f(Z_{i};\theta)$,
where $\theta=(\alpha,\beta_{0},\beta_{1},\eta_{0},\eta_{1},\lambda)$
and 
\begin{equation}
	f(Z_{i};\theta)=\pr(R_{i}\mid A_{i},X_{i};\eta_{A_{i}})f(Y_{i}\mid A_{i},X_{i};\beta_{A_{i}})\pr(A_{i}\mid X_{i};\alpha)f(X_{i};\lambda).
	\label{eq::parametricmodels}
\end{equation}
The observed-data likelihood is $L(\theta\mid Z_{\obsR,1},\ldots,Z_{\obsR,n})=\prod_{i=1}^{n}\{\sum_{r\in\mathcal{R}}I(R_i=r) \int f(Z_{i};\theta)\de\nu(X_{\mis,i})\}$.
Under Assumptions \ref{asump:OIMissing}\textendash \ref{asump:bddcomp}
as in Theorem \ref{thm2}, $\theta$ is identifiable if the parametric
models in \eqref{eq::parametricmodels} are not over-parametrized. The bounded completeness condition
holds for many commonly-used models, such as generalized linear models,
a location family of absolutely continuous distributions with a compact
support, and so on. See \citet{blundell2007semi}, \citet{hu2017nonparametric},
and the Supplementary Material for additional examples.

We first discuss the likelihood-based inference. Let $\tau(X_{i};\theta)=E(Y_{i}\mid A_{i}=1,X_{i};\beta_{1})-E(Y_{i}\mid A_{i}=1,X_{i};\beta_{0})$ be the covariate-specific average causal effect,
and let
\[
\hat{\tau}(\theta)=n^{-1}\sum_{i=1}^{n}\tau(X_{i};\theta),\quad\tau=\tau(\theta)=E\{\tau(X_{i};\theta)\}=E\{\hat{\tau}(\theta)\}.
\]
We first obtain the maximum likelihood estimate $\hat{\theta}$ and
then estimate $\tau$ by $\tau(\hat{\theta})$. The formula $\tau(\theta)$
involves integrating over the distribution of the confounders. To avoid
this complexity, we can use $\hat{\tau}(\hat{\theta})$ to estimate
$\tau$. We can use the bootstrap to construct confidence intervals.

We then discuss the Bayesian inference. Suppose that we can simulate
the posterior distribution of the missing confounders and the parameter
$\theta$. They further induce posterior distributions of $\hat{\tau}(\theta)$
and $\tau=\tau(\theta)$. Technically, the posterior distribution
of $\hat{\tau}(\theta)$ is different from that of $\tau$. The former
depends on the observed confounder values, but the latter does not.
See \citet{ding2018causal} for more discussions.

We give more computational details in the Supplementary Material,
including a fractional imputation algorithm \citep{yang2015fractional}
and a Bayesian procedure for a parametric model. Our future work will develop multiple imputation methods under Assumptions \ref{asump:OIMissing}\textendash \ref{asump:bddcomp}. From \eqref{eq::parametricmodels}, we need to use both the treatment and outcome models in the imputation step as in the full Bayesian procedure.

\section{Simulation\label{sec:Simulation-study}}

\subsection{Design of the simulation}

We use simulation to compare our estimators to existing ones. First, we consider the
unadjusted estimator, which is the simple difference-in-means of the
outcomes between the treated and control groups. We use it to quantify
the degree of confounding. Second, we consider the generalized
propensity score weighting estimator, with the generalized propensity
scores estimated separately by logistic regressions for each missing pattern \citep{rosenbaum1984reducing}. Third, we consider three multiple imputation estimators. The first estimator uses the outcome in the
imputation model, but the second does not \citep{mitra2011estimating}.
The third estimator uses the missing pattern in the propensity score
model \citep{qu2009propensity}.

We evaluate the finite sample performance of these estimators with
the missingness of confounders satisfying Assumption \ref{asump:OIMissing}.
In the first setting in \S \ref{subsec:Confounder}, one confounder has missing values, and we investigate the performance of the proposed nonparametric
estimator and the sensitivity to the choice of tuning parameters.
In the second setting in \S \ref{subsec:Multiple X}, multiple confounders have missing values, and we investigate the performance of the proposed parametric
estimator. In each setting, we choose sample size $n=400$,
$800$ and $1600$, and generate $2,000$ Monte Carlo samples for each sample size. For
multiple imputation estimators, we generate $100$ imputed datasets.
For all estimators, we use the bootstrap with $500$ replicates to
estimate the variances.

\subsection{One confounder subject to missingness\label{subsec:Confounder}}

The confounders $X_{i}=(X_{1i},X_{2i})$ follow $X_{1i}\sim\mathcal{N}(1,1)$
and $X_{2i}\sim$ Bernoulli$(0.5)$. The potential outcomes
follow $Y_{i}(0)=0.5+2X_{1i}+X_{2i}+\epsilon_{0i}$ and $Y_{i}(1)=3X_{1i}+2X_{2i}+\epsilon_{1i}$,
where $\epsilon_{0i}\sim\mathcal{N}(0,1)$ and $\epsilon_{1i}\sim\mathcal{N}(0,1)$.
The average causal effect $\tau$ is $1.$ The treatment indicator
$A_{i}$ follows Bernoulli($\pi_{i})$, where logit$(\pi_{i})=1.25-0.5X_{1i}-0.5X_{2i}$.
The missing indicator of $X_{1i},$ $R_{1i}$, follows Bernoulli$(p_{i})$,
where logit$(p_{i})=-2+2X_{1i}+A_{i}(1.5+X_{2i})$. The average response
rate is about $67\%$. Other variables do not have missing values.

For the proposed nonparametric estimator, we estimate $\hat{\tau}(X)$
using cubic splines with $5$ knots, and the density functions using
kernel-based estimators with the Gaussian kernel. We use the
$10$-fold cross-validation to choose the smoothing parameters in the smoothing
spline estimator and the bandwidths in the kernel-based estimators.
For $\hat{\xi}_{ra}(X)$, we choose $J = 5$ Hermite polynomial basis
functions, and $B = 50$ as the bound for regularization.

Table \ref{tab:Simulation-results}(a)
compares the nonparametric estimator to the existing ones.
The unadjusted estimator, the propensity score weighting estimator,
and multiple imputation estimators are biased. As a result, the coverage
rates of the confidence intervals for these methods are quite poor.
%Multiple imputation has the worst performance when the imputation
%model uses the outcome. 
Our proposed method has negligible biases
and good coverages, with variances decreasing with the sample sizes.

\begin{table}[t]
	\caption{\label{tab:Simulation-results} Simulation: bias ($\times10^{-2}$)
		and variance ($\times10^{-3}$) of the point estimator of $\tau$,
		variance estimate ($\times10^{-3}$), and coverage ($\%$) of $95\%$
		confidence intervals}
	
	\resizebox{\textwidth}{!}{%
		\begin{tabular}{lcccccccccccc}
			Method  & Bias  & Var  & VE & Cvg  & Bias  & Var  & VE & Cvg  & Bias  & Var  & VE & Cvg\tabularnewline
			\multicolumn{13}{c}{(a) Comparing the nonparametric estimator with existing ones}\tabularnewline
			& \multicolumn{4}{c}{$n=400$} & \multicolumn{4}{c}{$n=800$} & \multicolumn{4}{c}{$n=1600$}\tabularnewline
			Unadj  & $-127.5$  & $77.4$  & $73.7$ & $0.3$  & $-127.4$  & $38.0$  & $37.5$ & $0.0$  & $-127.2$  & $17.5$  & $18.6$ & $0.0$\tabularnewline
			GPSW  & $-55.1$  & $42.4$  & $44.2$ & $22.2$  & $-54.9$  & $20.9$  & $20.7$ & $5.8$  & $-54.4$  & $9.5$  & $9.9$ & $0.4$\tabularnewline
			MI1  & $41.5$  & $35.4$  & $36.7$ & $40.6$  & $41.0$  & $15.5$  & $17.2$ & $9.5$  & $40.8$  & $7.6$  & $8.3$ & $0.5$\tabularnewline
			MI2  & $-10.8$  & $60.0$  & $63.8$ & $91.4$  & $-9.2$  & $28.8$  & $30.8$ & $91.4$  & $-9.1$  & $13.7$  & $14.9$ & $86.6$\tabularnewline
			MIMP  & $29.3$  & $73.5$  & $71.5$ & $83.7$  & $28.5$  & $33.7$  & $32.6$ & $65.0$  & $28.3$  & $14.9$  & $16.0$ & $30.8$\tabularnewline
			NonPara  & $1.2$  & $19.4$  & $18.8$ & $95.1$  & $0.9$  & $9.6$  & $8.1$ & $95.2$  & $0.8$  & $3.9$  & $3.8$ & $94.9$\tabularnewline
			\multicolumn{13}{c}{(b) Comparing the parametric estimator with existing ones}\tabularnewline
			& \multicolumn{4}{c}{$n=400$} & \multicolumn{4}{c}{$n=800$} & \multicolumn{4}{c}{$n=1600$}\tabularnewline
			Unadj  & $32.2$ & $85.2$ & $85.8$ & $81.5$ & $32.2$ & $44.3$ & $42.9$ & $65.8$ & $31.9$ & $20.3$ & $21.6$ & $43.1$\tabularnewline
			GPSW  & $8.4$ & $174.6$ & $246.1$ & $97.2$ & $8.8$ & $84.2$ & $94.2$ & $94.9$ & $8.3$ & $40.0$ & $44.0$ & $92.4$\tabularnewline
			MI1  & $7.7$ & $180.5$ & $238.0$ & $96.1$ & $7.1$ & $93.5$ & $106.4$ & $95.2$ & $6.9$ & $47.5$ & $54.8$ & $93.4$\tabularnewline
			MI2  & $3.0$ & $162.1$ & $209.9$ & $97.3$ & $3.1$ & $84.2$ & $94.1$ & $95.8$ & $2.6$ & $42.8$ & $49.1$ & $94.6$\tabularnewline
			MIMP  & $12.9$ & $177.0$ & $239.2$ & $95.7$ & $12.2$ & $93.9$ & $107.5$ & $93.8$ & $12.1$ & $47.4$ & $55.0$ & $91.8$\tabularnewline
			Para  & $1.6$ & $95.4$ & $95.4$ & $95.3$ & $0.4$ & $48.3$ & $48.0$ & $95.0$ & $0.0$ & $23.0$ & $24.2$ & $95.4$\tabularnewline
	\end{tabular}}
	
	{\footnotesize{}{}{}{}{}{}{}Unadj: the unadjusted estimator;
		GPSW: the generalized propensity score weighting estimator; NonPara: the proposed nonparametric estimator;
		Para: the proposed parametric estimator; For the multiple imputation estimators,
		MI1 uses the outcome in the imputation, MI2 does not use the outcome
		in the imputation, and MIMP is the multiple imputation missingness
		pattern method of \citet{qu2009propensity}. }{\footnotesize \par}
\end{table}

%We assess the sensitivity of our nonparametric estimator to the choice
%of tuning parameters $J$ and $B$ in the Supplementary Material.
%Below we give more details for our simulation in \S \ref{subsec:Confounder}.

To assess the sensitivity of the nonparametric estimator to the choice
of tuning parameters $J$ and $B$, we specify a $4\times3$ design
with $(J,B)\in\{(3,50),(3,100),(5,50),(5,100)\}$ and $n\in\{400,800,1600\}$.
Table \ref{tab:addSim} shows the mean squared errors. For each $(J,B)$, the mean squared error decreases
with the sample size. The mean squared error decreases with $J$,
and is relatively insensitive to the choice of $B$. The mean squared
errors remain small across all cases.

\begin{table}[t]
	\caption{\label{tab:addSim}Simulation results for different tuning parameters:
		mean squared errors ($\times10^{-3}$) of the proposed estimator of
		$\tau$ for different choices of $(J,B)$ based on $2,000$ Monte
		Carlo samples}
	\centering{}%
	\begin{tabular}{lccc}
		$(J,B)$  & \multicolumn{1}{c}{$n=400$} & \multicolumn{1}{c}{$n=800$} & \multicolumn{1}{c}{$n=1600$}\tabularnewline
		$(3,50)$  & $26.8$  & $13.9$  & $8.3$\tabularnewline
		$(3,100)$  & $27.0$  & $14.1$  & $8.7$\tabularnewline
		$(5,50)$  & $19.5$  & $9.7$  & $4.1$\tabularnewline
		$(5,100)$  & $21.3$  & $10.2$  & $4.5$\tabularnewline
	\end{tabular}
\end{table}

\subsection{Multiple confounders subject to missingness\label{subsec:Multiple X}}

Let $X_{i}=(X_{1i},\ldots,X_{6i})$. We generate $X_{1i}$ and $X_{2i}$
from $\mathcal{N}(1,1)$, $X_{3i}$ and $X_{4i}$ from \{Bernoulli$(0.5)-0.5\}/0.5$,
$X_{5i}=X_{1i}+X_{2i}+X_{3i}+X_{4i}+\epsilon_{5i}$ with $\epsilon_{5i}\sim\mathcal{N}(0,1)$,
and $X_{6i}$ from Bernoulli$(p_{6i})$ with logit$(p_{6i})=-X_{5i}$.
The potential outcomes follow $Y_{i}(0)=(1,X_{i}^{\T})\beta_{0}+\epsilon_{0i}$
and $Y_{i}(1)=(1,X_{i}^{\T})\beta_{1}+\epsilon_{1i}$, where $\beta_{0}=(-1.5,1,-1,1,-1,1,1)^{\T}$
and $\beta_{1}=(0,-1,1,-1,1,-1,-1)^{\T}$, $\epsilon_{0i}\sim\mathcal{N}(0,1)$
and $\epsilon_{1i}\sim\mathcal{N}(0,1)$. The average treatment
effect is $\tau=-0.5$. The treatment indicator $A_{i}$ follows Bernoulli($\pi_{i})$,
where logit$(\pi_{i})=(1,X_{i}^{\T})\alpha$ and $\alpha=0.5\times(2,1,1,1,1,-2,-2)^{\T}$.
Covariates $X_{5i}$ and $X_{6i}$ have missing values, but other variables
do not. The missingness pattern for $X_{5i}$ and $X_{6i}$, $R_{i}=(R_{5i},R_{6i})\in\{(11),(10),(01),(00)\}$,
follows Multinominal$(p_{11,i},p_{10,i},p_{01,i},p_{00,i})$, where
logit($p_{11,i}$)$=[1+3\exp\{(1,A_{i},X_{i}^{\T})\eta\}]^{-1}$,
logit($p_{kl,i}$)$=[\exp\{-(1,A_{i},X_{i}^{\T})\eta\}+3]^{-1}$ for
$kl\in\{10,01,00\}$, and $\eta=0.25\times(-4,1,1,1,1,1,-1,-1)^{\T}$.
The average percentages of these missingness patterns are about $49\%$, $17\%$, $17\%$ and $17\%$,
respectively.

%We compare five estimators specified in Table \ref{tab:Simulation-results}.
%In the proposed method, because $f(Y\mid A=a,X)$ is a Gaussian distribution,
%it is easy to verify that $f(A=a,X,Y,R_{X}=1_{2})$ is the bounded
%complete in $Y$; see Proposition \ref{eg:Gaussian}. For the proposed
%estimator, we consider the true parametric models and the EM algorithm
%similar to that in $\mathsection$\ref{subsec:Frequentist-perspective},
%where the proposal distributions for the missing values of $X_{5i}$
%and $X_{6i}$ is the normal and Bernoulli distributions with the moments
%being the sample moments of the observed $X_{5i}$ and $X_{6i}$,
%respectively.

Table \ref{tab:Simulation-results}(b) compares the parametric maximum
likelihood estimator to the existing ones. The unadjusted estimator
has large biases due to confounding. Multiple
imputation estimators have large biases, although the coverages of
confidence intervals appear good due to the overestimation of variances.
In contrast, our estimator has negligible biases and good coverages.

\section{Applications\label{sec:Application}}

\subsection{The causal effect of smoking on the blood lead level\label{sec:The-causal-effect}}

We use a dataset from the 2015\textendash 2016 U.S. National
Health and Nutrition Examination Survey to estimate the causal effect
of smoking on the blood lead level \citep{hsu2013calibrating}. The data set includes $2949$ adults consisting of $1102$ smokers, denoted by $A=1$, and $1847$
nonsmokers, denoted by $A=0$. All subjects were at least 15 years old and had no tobacco use besides cigarette smoking in the previous 5 days.
The outcome $Y$ is the lead level in blood, ranging from $0.05$ ug/dl to $23.51$
ug/dl. The confounders $X$ include the income-to-poverty level, age
and gender. The income-to-poverty level has missing values, but
other variables do not have. The missingness of the income-to-poverty
level is likely to be not at random because subjects with high incomes
may be less likely to disclose their income information. It is plausible that Assumption \ref{asump:OIMissing} holds, because this missingness is perceivably unrelated to the lead level in blood after controlling for the income information \citep{davern2005effect}. The missing rate of
the income-to-poverty is $14.0\%$ for smokers and $15.2\%$ for non-smokers.
We apply the proposed procedure to obtain estimates
separately for groups stratified by age and gender, and then average
over the empirical distribution of age and gender.

Table \ref{tab:applications}(a) shows the results. We note substantial differences in the point estimates
between our estimator and the competitors, which illustrate the impact
of the missing data assumption on causal inference in the presence
of missing confounders. In contrast to the existing estimators, our
estimator handles the confounders missing not at random more properly.
Based on the nonparametric estimator, smoking increases the lead level
in blood by $0.20$ ug/dl on average.

\begin{table}
	\caption{\label{tab:applications} Point estimate, standard error
		by the bootstrap and $95\%$ confidence interval}
	
	\begin{tabular}{ccccccccc}
		& Est  & SE  & $95\%$ CI  &  &  & Est  & SE  & $95\%$ CI\tabularnewline
		\multicolumn{9}{c}{(a) The causal effect of smoking on the blood lead level in \S\ref{sec:The-causal-effect}}\tabularnewline
		Unadj    & $0.44$  & $0.05$  & $(0.35,0.54)$  &  & MI1  & $0.34$  & $0.05$  & $(0.25,0.44)$ \tabularnewline
		PSW      & $0.12$  & $0.05$  & $(0.02,0.22)$  &  & MI2  & $0.35$  & $0.05$  & $(0.25,0.44)$\tabularnewline
		NonPara  & $0.20$  & $0.07$  & $(0.05,0.36)$  &  & MIMP & $0.35$  & $0.05$  & $(0.25,0.44)$\tabularnewline
		\multicolumn{9}{c}{(b) The causal effect of education on general health satisfaction in \S\ref{sec::health}}\tabularnewline
		Unadj  & $-0.57$  & $0.034$  & $(-0.64,-0.51)$  &  & MI1  & $-0.24$  & $0.057$  & $(-0.36,-0.13)$\tabularnewline
		GPSW  & $-0.25$  & $0.054$  & $(-0.36,-0.14)$  &  & MI2  & $-0.26$  & $0.057$  & $(-0.38,-0.15)$\tabularnewline
		Para  & $-0.32$  & $0.051$  & $(-0.41,-0.21)$  &  & MIMP  & $-0.23$  & $0.057$  & $(-0.34,-0.11)$\tabularnewline
	\end{tabular}
	
	{\footnotesize{}{}{}{}{}Unadj: the unadjusted estimator; GPSW:
		the generalized propensity score weighting estimator; NonPara: the proposed nonparametric estimator;
		Para: the proposed parametric estimator; For the multiple imputation estimators,
		MI1 uses the outcome in the imputation, MI2 does not use the outcome
		in the imputation, and MIMP is the multiple imputation missingness
		pattern method of \citet{qu2009propensity}. }{\footnotesize \par}
\end{table}

\subsection{The causal effect of education on general health satisfaction}\label{sec::health}

We use a dataset from the 2015\textendash 2016 U.S. National Health
and Nutrition Examination Survey to estimate the average causal effect
of education on general health satisfaction. The dataset includes
$4,845$ subjects. Among them, $76\%$ individuals have at least
high school education, denoted by $A=1$, and $24\%$ do not, denoted
by $A=0$. The outcome $Y$ is the general health satisfaction
score ranging from $1$ to $5$, with lower values indicating better satisfaction.
The observed outcomes have mean $2.88$ and standard deviation $0.96$.
The confounders $X $ include age, gender, race, marital status, income-to-poverty level, and an indicator of ever having pre-diabetes risk.
The income-to-poverty level and pre-diabetes risk variables have missing
values, but other variables do not have. %The average percentages of the missingness
%patterns $(11)$, $(10)$, $(01)$ and $(00)$ are about $74\%$,
%$8\%$, $16\%$ and $2\%$, respectively. 
The missingness of the family poverty ratio and the pre-diabetes risk
variables is likely to be related to the missing values themselves.
It is plausible that this missingness is unrelated to the outcome
value conditioning on the treatment and confounders.

Table \ref{tab:applications}(b) shows the results. Although 
%quantitatively
qualitatively
all estimators show that education is beneficial in improving general
health satisfaction, we note differences in the point
estimates between our estimator and the competitors. This illustrates
the impact of the missing data assumption on causal inference with missing confounders. Based on the parametric estimator,
education improves the general health satisfaction by $0.32$ on average.

%\begin{table}
%	\caption{\label{tab:Results-2}Point estimate, standard error by the bootstrap and $95\%$ confidence interval}
%	
%	\centering{}%
%	\begin{tabular}{ccccccccc}
%		& Est  & SE  & $95\%$ CI  &  &  & Est  & SE  & $95\%$ CI\tabularnewline
%		Unadj  & $-0.57$ & $0.034$ & $(-0.64,-0.51)$ &  & MI1  & $-0.24$ & $0.057$ & $(-0.36,-0.13)$\tabularnewline
%		GPSW & $-0.25$ & $0.054$ & $(-0.36,-0.14)$ &  & MI2  & $-0.26$ & $0.057$ & $(-0.38,-0.15)$\tabularnewline
%		Proposed  & $-0.32$ & $0.051$ & $(-0.41,-0.21)$ &  & MIMP  & $-0.23$ & $0.057$ & $(-0.34,-0.11)$\tabularnewline
%	\end{tabular}
%\end{table}

\section*{Acknowledgments}

Yang is supported in part by Oak Ridge Associated Universities, U.S. National Science Foundation and National Institute of Health. 
Ding is supported in part by the U.S. National Science Foundation and Institute of Education Sciences. 
The authors thank Professor Eric Tchetgen Tchetgen for valuable discussions,
and Professor Xiaohong Chen for useful references. The comments from
the Associate Editor and two reviewers improved the quality of our
paper. 

\section*{Supplementary material}

Supplementary material includes additional
proofs, further discussions on the nonparametric and parametric
estimators, and additional simulation.

\bibliographystyle{dcu}
\bibliography{CI2_JOE}

\newpage 

\part*{Supplementary material}

\global\long\def\theequation{S\arabic{equation}}
\setcounter{equation}{0}

\global\long\def\thelemma{S\arabic{lemma}}
\setcounter{lemma}{0}

\global\long\def\theexample{S\arabic{example}}
\setcounter{example}{0}

\global\long\def\thesection{S\arabic{section}}
\setcounter{section}{0}

\global\long\def\thetheorem{S\arabic{theorem}}
\setcounter{theorem}{0}

\global\long\def\theremark{S\arabic{remark}}
\setcounter{remark}{0}

\global\long\def\thestep{S\arabic{step}}
\setcounter{step}{0}

\global\long\def\theassumption{S\arabic{assumption}}
\setcounter{assumption}{0}

\newtheorem{proof1}{Proof}
\global\long\def\theproof{S\arabic{proof1}}
\setcounter{proof1}{0}

\global\long\def\theproposition{S\arabic{proposition}}
\setcounter{proposition}{0}

\global\long\def\thefigure{S\arabic{figure}}
\setcounter{proposition}{0}

\global\long\def\thefigure{S\arabic{table}}
\setcounter{table}{0}

\section{Proofs}

\subsection{Proof of Proposition \ref{proposition1}}

We prove the result for $p=2.$ Proofs for general $p$ are similar
and hence omitted. For discrete covariates with $R=(0,0)$, (\ref{eq:m2})
reduces to 
\begin{eqnarray}
f\{A=a,Y,R=(0,0)\} & = & \sum_{i=1}^{J_{1}}\sum_{j=1}^{J_{2}}\frac{\pr\{R=(0,0)\mid X_{1i},X_{2j},A=a\}}{\pr\{R=(1,1)\mid X_{1i},X_{2j},A=a\}}\nonumber \\
&  & \times f\{A=a,X_{1i},X_{2j},Y,R=(1,1)\},\quad(a=0,1).\label{eq:1}
\end{eqnarray}
In a matrix form, (\ref{eq:1}) becomes 
\begin{equation}
\left(\begin{array}{c}
f\{A=a,y_{1},R=(0,0)\}\\
\vdots\\
f\{A=a,y_{K},R=(0,0)\}
\end{array}\right)_{K\times1}=\Theta_{a}\left(\begin{array}{c}
\xi_{(0,0)a}(x_{11},x_{21})\\
\vdots\\
\xi_{(0,0)a}(x_{1J_{1}},x_{2J_{2}})
\end{array}\right)_{(J_{1}J_{2})\times1},\label{eq:p3}
\end{equation}
where 
\[
\Theta_{a}=\left(\begin{array}{ccc}
f\{A=a,x_{11},x_{21},y_{1},R=(1,1)\} & \cdots & f\{A=a,x_{1J_{1}},x_{2J_{2}},y_{1},R=(1,1)\}\\
\vdots & \ddots & \vdots\\
f\{A=a,x_{11},x_{21},y_{K},R=(1,1)\} & \cdots & f\{A=a,x_{1J_{1}},x_{2J_{2}},y_{K},R=(1,1)\}
\end{array}\right)_{K\times(J_{1}J_{2})},
\]
and 
\[
\xi_{(0,0)a}(x_{1i},x_{2j})=\frac{\pr\{R=(0,0)\mid A=a,x_{1i},x_{2j}\}}{\pr\{R=(1,1)\mid A=a,x_{1i},x_{2j}\}}.
\]
In the linear system (\ref{eq:p3}), the vector on the left hand side
and the coefficients in $\Theta_{a}$ on the right hand side are identifiable
because they depend only on the observed data. The linear system for
the $\xi_{(0,0)a}(X_{1},X_{2})$'s has a unique solution if and only
if $\Theta_{a}$ has a full column rank $J_{1}J_{2}$. Similarly,
for $R=(1,0)$,

\begin{equation}
\left(\begin{array}{c}
f\{A=a,X_{1},y_{1},R=(1,0)\}\\
\vdots\\
f\{A=a,X_{1},y_{K},R=(1,0)\}
\end{array}\right)_{K\times1}=\Theta_{X_{1}a}\left(\begin{array}{c}
\xi_{(1,0)a}(X_{1},x_{21})\\
\vdots\\
\xi_{(1,0)a}(X_{1},x_{2J_{2}})
\end{array}\right)_{J_{2}\times1},\quad(a=0,1),\label{eq:p4}
\end{equation}
where 
\[
\Theta_{X_{1}a}=\left(\begin{array}{ccc}
f\{A=a,X_{1},x_{21},y_{1},R=(1,1)\} & \cdots & f\{A=a,X_{1},x_{2J_{2}},y_{1},R=(1,1)\}\\
\vdots & \ddots & \vdots\\
f\{A=a,X_{1},x_{21},y_{K},R=(1,1)\} & \cdots & f\{A=a,X_{1},x_{2J_{2}},y_{K},R=(1,1)\}
\end{array}\right)_{K\times J_{2}}.
\]
The linear system (\ref{eq:p4}) has a unique solution for the $\xi_{(1,0)a}(X_{1},X_{2})$'s
if and only if $\Theta_{X_{1}a}$ has a column rank $J_{2}$, which
is guaranteed if $\Theta_{a}$ has a full column rank $J_{1}J_{2}$
. For $R=(0,1)$,

\begin{equation}
\left(\begin{array}{c}
f\{A=a,X_{2},y_{1},R=(0,1)\}\\
\vdots\\
f\{A=a,X_{2},y_{K},R=(0,1)\}
\end{array}\right)_{K\times1}=\Theta_{x_{2}a}\left(\begin{array}{c}
\xi_{(0,1)a}(x_{11},X_{2})\\
\vdots\\
\xi_{(0,1)a}(x_{1J_{1}},X_{2})
\end{array}\right)_{J_{1}\times1},\quad(a=0,1),\label{eq:p5}
\end{equation}
where 
\[
\Theta_{X_{2}a}=\left(\begin{array}{ccc}
f\{A=a,x_{11},X_{2},y_{1},R=(1,1)\} & \cdots & f\{A=a,x_{1J_{1}},X_{2},y_{1},R=(1,1)\}\\
\vdots & \ddots & \vdots\\
f\{A=a,x_{11},X_{2},y_{K},R=(1,1)\} & \cdots & f\{A=a,x_{1J_{1}},X_{2},y_{K},R=(1,1)\}
\end{array}\right)_{K\times J_{1}}.
\]
The linear system (\ref{eq:p5}) has a unique solution for the $\xi_{(0,1)a}(X_{1},X_{2})$'s
if and only if $\Theta_{X_{2}a}$ has a column rank $J_{1}$, which
is guaranteed if $\Theta_{a}$ has a full column rank $J_{1}J_{2}$.
Therefore, $\xi_{ra}(X_{1},X_{2})$ is identifiable if and only if
$\Theta_{a}$ has a full column rank $J_{1}J_{2}$.

It follows that 
\[
\pr(R=r\mid A=a,X_{1},X_{2})=\frac{\xi_{ra}(X_{1},X_{2})}{\sum_{r'\in\mathcal{R}}\xi_{r'a}(X_{1},X_{2})}
\]
is identifiable. It then follows that 
\[
f(A=a,X,Y)=\frac{f(A=a,X,Y,R=1_{p})}{\pr(R=1_{p}\mid A=a,X_{1},X_{2})}
\]
is identifiable. Therefore, the joint distribution of $(A,X,Y,R)$,
$f(A=a,X,Y)\pr(R=r\mid A=a,X)$, is identifiable. This completes the
proof.

\subsection{Proof of Remark \ref{rmk: bddcomp}}

We first prove that when $X$ and $Y$ are discrete with finite supports,
Assumption \ref{asump:bddcomp} is equivalent to the rank condition
in Proposition \ref{proposition1}.

\begin{proposition}\label{prop:equi}
	
	Suppose that $X$ and $Y$ are discrete, and that $X_{j}\in\{x_{j1},\ldots x_{jJ_{j}}\}$
	for $j=1,\ldots,p$ and $Y\in\{y_{1},\ldots,y_{K}\}$. The bounded
	completeness in $Y$ of $f(A=a,X,Y,R=1_{p})$ is equivalent to the
	condition that $\Theta_{a}$ is of full column rank, for $a=0,1$.
	
\end{proposition}

\begin{proof} Suppose that $\int g(X)f(A=a,X,Y,R=1_{p})\de\nu(X)=0$
	for all $Y=y_{1},\ldots,y_{K}$. For discrete $X$, the integral equation
	(\ref{eq:m2}) reduces to 
	\begin{equation}
	\Theta_{a}\left(\begin{array}{c}
	g(x_{11},\ldots,x_{p1})\\
	\vdots\\
	g(x_{1J_{1}},\ldots,x_{pJ_{p}})
	\end{array}\right)_{(J_{1}\times\cdots\times J_{p})\times1}=\left(\begin{array}{c}
	0\\
	\vdots\\
	0
	\end{array}\right)_{K\times1}.\label{eq:linear sys}
	\end{equation}
	If $\Theta_{a}$ is of full column rank, then the solution to the
	linear system (\ref{eq:linear sys}) is zero, that is, $g(X)=0$,
	which indicates that $f(A=a,X,Y,R=1_{p})$ is bounded complete in
	$Y$.
	
	On the other hand, suppose $f(A=a,X,Y,R=1_{p})$ is bounded complete
	in $Y$. Therefore, $\int g(X)f(A=a,X,Y,R=1_{p})\de\nu(X)=0$ for
	all $Y=y_{1},\ldots,y_{K}$ implies $g(X)=0$. In this case, the only
	solution to (\ref{eq:linear sys}) is
	
	\[
	\left(\begin{array}{c}
	g(x_{11},\ldots,x_{p1})\\
	\vdots\\
	g(x_{1J_{1}},\ldots,x_{pJ_{p}})
	\end{array}\right)_{(J_{1}\times\cdots\times J_{p})\times1}=\left(\begin{array}{c}
	0\\
	\vdots\\
	0
	\end{array}\right)_{(J_{1}\times\cdots\times J_{p})\times1}.
	\]
	Therefore, $\Theta_{a}$ is of full column rank. This completes the
	proof. \end{proof}

We then prove that Assumption \ref{asump:bddcomp} implies Assumption
\ref{asump-overlap}.

\begin{proof} For the discrete $X$ and $Y$, suppose that there
	exists $x^{*}$ with $\pr(X=x^{*})>0$, such that $e(x^{*})=\pr(A=1\mid X=x^{*})=0$.
	Then, 
	\[
	f(A=1,X=x^{*},Y,R=1_{p})=e(x^{*})f(X=x^{*},Y)\pr(R=1_{p}\mid X=x^{*},Y,A=1)=0,
	\]
	which indicates that one column in $\Theta_{1}$ is zero. Therefore,
	$\Theta_{1}$ is not of full column rank, violating the bounded completeness
	condition.
	
	For the continuous $X$ and $Y$, suppose that there exists a subset
	$\mathcal{X}^{*}$ with $\pr(x^{*}\in\mathcal{X}^{*})>0$, such that
	$e(x^{*})=\pr(A=1\mid X=x^{*})=0$ for any $x^{*}\in\mathcal{X}^{*}$.
	Following the same derivation as for the discrete case, we have, for
	any $x^{*}\in\mathcal{X}^{*}$, that $f(A=1,X=x^{*},Y,R=1_{p})=0$.
	Then, $f(A=1,X,Y,R=1_{p})$ is not bounded complete in $Y$. To see
	this, suppose $\int g(X)f(A=1,X,Y,R=1_{p})\de\nu(X)=0$ for any $Y$,
	we can let $g(X)$ be zero outside of $\mathcal{X}^{*}$ but non-zero
	inside of $\mathcal{X}^{*}$, violating the bounded completeness condition.
\end{proof}

\section{More details for the nonparametric estimation of $\tau$\label{sec:Nonparametric-estimation-1}}

\subsection{Regularization of series estimators\label{sec:Regularization-of-series}}

Although we can use other regularization techniques to solve the ill-conditioned
inverse problem such as Tikhonov's regularization \citep{darolles2011nonparametric}
and a penalized sieve minimum distance criterion \citep{chen2015sieve},
we follow \citet{newey2003instrumental} to restrict $\xi_{ra}(X)$
and its estimator $\hat{\xi}_{ra}(X)$ to belong to a compact space.
Because the inverse of integration restricted to a compact space is
continuous, this regularization turns the problem to be well-posed.

We now describe the compact space and its norm. Recall that $p$ is
the dimension of $X$. \textcolor{black}{For any function $g(X)$,
	denote 
	\[
	D^{\lambda}g(X)=\frac{\partial^{\lambda_{1}}}{\partial x_{1}^{\lambda_{1}}}\cdots\frac{\partial^{\lambda_{p}}}{\partial x_{p}^{\lambda_{p}}}g(X),
	\]
	and $|\lambda|=\sum_{l=1}^{p}\lambda_{l}$ gives the order of the
	derivative. In particular, the zero order derivative is the function
	itself; that is, $D^{0}g(X)=g(X)$.} For $H(X)=\{h_{1}(X),\ldots,h_{J}(X)\}$,
\textcolor{black}{we define $D^{\lambda}H(X)=\{D^{\lambda}h_{1}(X),\ldots,D^{\lambda}h_{J}(X)\}^{\T}$.
}Let $\lambda=(\lambda_{1},\ldots,\lambda_{p})^{\T}$ be a $p$-vector
with non-negative integers as components. For $m>0$, $m_{0},\delta_{0}>p/2$,
and $p/2<\delta<\delta_{0}$, consider the following functional space
\begin{equation}
\mathcal{G}_{m,m_{0},\delta,B}=\left\{ g(X):\sum_{|\lambda|\leq m+m_{0}}\int\{D^{\lambda}g(\tilde{X})\}^{2}(1+\tilde{X}^{\T}\tilde{X})^{\delta_{0}}dx\leq B\right\} ,\label{eq:G}
\end{equation}
where $\tilde{X}$ is the standardized version of $X$ . Consider
the norm 
\[
||g||_{\mathcal{G}}=\max_{|\lambda|\leq m}\sup_{X}|D^{\lambda}g(\tilde{X})|(1+\tilde{X}^{\T}\tilde{X})^{\delta}.
\]
\citet{gallant1987semi} showed that the closure of $\mathcal{G}_{m,m_{0},\delta_{0},B}$
with respect to the norm $||g||_{\mathcal{G}}$ is compact.

\begin{assumption}[Regularization of the parameter space]\label{assump:regularization}
	Assume that $\xi_{ra}(X)$ and its estimator $\hat{\xi}_{ra}(X)$
	belong to $\mathcal{G}_{m,m_{0},\delta_{0},B}$ in (\ref{eq:G}),
	for any $r$ and $a$.
	
\end{assumption}

\begin{remark}
	
	The regularization is not restrictive for the following reasons. First,
	by the definition of $\mathcal{G}_{m,m_{0},\delta_{0},B}$, the bound
	$B$ requires the functions of $\mathcal{G}_{m,m_{0},\delta_{0},B}$
	to be smooth to a certain degree and the tails of these functions
	to be small. In most applications, we would expect that the functions
	$\xi_{ra}(X)$ to be smooth and mainly concerned with the functional
	forms of $\xi_{ra}(X)$ over some compact region that is large enough
	to cover the region where observations are measured.
	
\end{remark}

Given the Hermite approximation of $\xi_{ra}(X)$, the regularization
in Assumption \ref{assump:regularization} becomes 
\begin{equation}
\beta_{ra}^{\T}\left[\sum_{|\lambda|\leq m+m_{0}}\int\{D^{\lambda}H(\tilde{X})\}\{D^{\lambda}H(\tilde{X})\}^{\T}(1+\tilde{X}^{\T}\tilde{X})^{\delta_{0}}\de\nu(X)\right]\beta_{ra}\leq B,\label{eq:restrict}
\end{equation}
where $\beta_{ra}=(\beta_{ra}^{1},\ldots,\beta_{ra}^{J})^{\T}$. Therefore,
we choose the positive definite matrix $\Lambda$ in the constraint
for regularization in $\mathsection$\ref{sec:Nonparametric-estimation}
to be

\[
\text{\ensuremath{\Lambda}=}\sum_{|\lambda|\leq m+m_{0}}\int\{D^{\lambda}H(\tilde{X})\}\{D^{\lambda}H(\tilde{X})\}^{\T}(1+\tilde{X}^{\T}\tilde{X})^{\delta_{0}}\de\nu(X).
\]

The proposed estimator of $\xi_{ra}$ is $\hat{\xi}_{ra}(X)=\sum_{j=1}^{J}\hat{\beta}_{ra}^{j}h^{j}(\tilde{X}),$
where $\hat{\beta}_{ra}$ minimizes (\ref{eq:obj fctn}) with the
constraint $\beta_{ra}^{\T}\Lambda\beta_{ra}\leq B$.

\subsection{The computational algorithm in $\mathsection$\ref{sec:Nonparametric-estimation}
	and an example\label{subsec:Computation-algorithm} }

We summarize the computation algorithm for $\tau$ as follows.

\begin{step}\label{step1} Obtain nonparametric estimators of $\tau(X)$, 
	$f(X\mid A=a,R=1_{p})$, $f(X_{\obs},Y\mid A=a,R=r)$, for all $r$
	and $a$. Specifically, we use 
	\begin{equation}
	\hat{\tau}(X)=\hat{E}(Y\mid A=1,X,R=1_{p})-\hat{E}(Y\mid A=0,X,R=1_{p}),\label{eq:tauhat-spline}
	\end{equation}
	where $\hat{E}(Y\mid A=a,X,R=1_{p})$ is a smoothing spline estimator
	of $E(Y\mid A=a,R=1_{p})$, for $a=0,1$. Also let $\hat{f}(X\mid A=a,R=1_{p})$
	and $\hat{f}(X_{\obs},Y\mid A=a,R=r)$ be the kernel density estimators
	of $f(X\mid A=a,R=1_{p})$ and $f(X_{\obs},Y\mid A=a,R=r)$, respectively.
	
\end{step}

\begin{step}\label{step2} Obtain a series estimator of $\xi_{ra}(X)$
	using the Hermite polynomials, $\hat{\xi}_{ra}(X)\approx\sum_{j=1}^{J}\hat{\beta}_{ra}^{j}h^{j}(\tilde{X})$,
	where $ (\hat{\beta}_{ra}^{1},\ldots,\hat{\beta}_{ra}^{J})^{\T}$
	minimizes (\ref{eq:obj fctn}) with the constraint $\beta_{ra}^{\T}\Lambda\beta_{ra}\leq B$.
	
\end{step}

\begin{step}\label{step3} Estimate the probabilities $\pr(R=1_{p}\mid A=a,X)$
	by $\widehat{\pr}(R=1_{p}\mid A=a,X)=\left\{ 1+\sum_{r\neq1_{p}}\hat{\xi}_{ra}(X)\right\} ^{-1}$.
	
\end{step}

\begin{step}\label{step4} Estimate $\tau$ by (\ref{eq:tauhat})
	using a numerical approximation.
	
\end{step}

For illustration of the proposed computational algorithm, we provide
an example with a scalar $X$, which is subject to the outcome-independent
missingness. In this case, $R\in\mathcal{R}=\{0,1\}$.

\begin{example} In Step \ref{step1}, obtain a nonparametric estimator
	of $\tau(X)$ as 
	\[
	\hat{\tau}(X)=\hat{E}(Y\mid A=1,X,R=1)-\hat{E}(Y\mid A=0,X,R=1),
	\]
	where $\hat{E}(Y\mid A=a,X,R=1)$ is a smoothing spline estimator
	of $E(Y\mid A=a,X,R=1)$, for $a=0,1$. Also let 
	\[
	\hat{f}(X\mid A=a,R=1),\quad\hat{f}(Y\mid A=a,R=0),\quad\hat{f}(Y,R=0\mid A=a),\quad\hat{f}(Y,R=1\mid A=a)
	\]
	be the kernel density estimators of 
	\[
	f(X\mid A=a,R=1),\quad f(Y\mid A=a,R=0),\quad f(Y,R=0\mid A=a),\quad f(Y,R=1\mid A=a).
	\]
	
	In Step \ref{step2}, (\ref{eq:m3}) becomes 
	\[
	\hat{f}(Y,R=0\mid A=a)=\int\xi_{0a}(X)\hat{f}(X\mid A=a,Y,R=1)\de\nu(X)\times\hat{f}(Y,R=1\mid A=a).
	\]
	Let $\hat{E}\{h^{j}(\tilde{X})\mid y,A=a,R=1\}$ be a nonparametric
	estimator of $E\{h^{j}(\tilde{X})\mid y,A=a,R=1\}$. For unit $i$,
	evaluate this nonparametric estimator at $Y_{i}$, we have $\hat{H}_{0a}^{j}=\hat{E}\{h^{j}(\tilde{X})\mid A_{i}=a,Y_{i},R_{i}=1\}$.
	We obtain a series estimator of $\xi_{0a}(X)$ using the Hermite polynomials,
	$\hat{\xi}_{0a}(X)\approx\sum_{j=1}^{J}\hat{\beta}_{0a}^{j}h^{j}(\tilde{X})$,
	where the $\hat{\beta}_{0a}^{j}$'s minimize the objective function 
	\begin{equation}
	\sum_{i=1}^{n}\left\{ \hat{f}(Y_{i},R_{i}=0\mid A_{i}=a)\vphantom{\sum_{j}^{J}}-\sum_{j=1}^{J}\beta_{0a}^{j}\hat{H}_{0a}^{j}\hat{f}(Y_{i},R_{i}=1\mid A_{i}=a)\right\} ^{2},\label{eq:obj fctn-S}
	\end{equation}
	subject to the constraint $\beta_{0a}^{\T}\Lambda\beta_{0a}\leq B$.
	
	In Step \ref{step3}, estimate the probability $\pr(R=1\mid A=a,X)$
	by $\widehat{\pr}(R=1\mid A=a,X)=\{1+\hat{\xi}_{0a}(X)\}^{-1}$.
	
	In Step \ref{step4}, obtain the estimator of $\tau$ by 
	\begin{equation}
	\sum_{a=0}^{1}\widehat{\pr}(A=a,R=1)\int\hat{\tau}(x)\frac{\hat{f}(X\mid A=a,R=1)}{\widehat{\pr}(R=1\mid A=a,X)}\de\nu(X),\label{eq:tauhat-1}
	\end{equation}
	using a numerical approximation.
	
\end{example}

\subsection{Choice of tuning parameters}

The proposed estimator depends on several tuning parameters: the number
of the Hermite polynomial functions $J$, the bound $B$ for regularization,
and tuning parameters in the kernel-based estimators. On the one hand,
$J$ and $B$ should be large enough to ensure that the series estimator
approximates the true underlying function well. On the other hand,
$J$ and $B$ should not be too large to control the variance
of our estimator. \citet{chen2012estimation} and \citet{chen2015optimal} investigated the general
requirements for these tuning parameters in terms of the growing rate
with the sample size in the penalized sieve minimum distance estimation.
In practice, we suggest using data-driven methods, such as cross-validation,
to choose these parameters, and conducting sensitivity analysis varying
the tuning parameters.

%Below we give more details for our simulation in \S \ref{subsec:Confounder}.
%To assess the sensitivity of our nonparametric estimator to the choice
%of tuning parameters $J$ and $B$, we specify a $4\times3$ design
%with $(J,B)\in\{(3,50),(3,100),(5,50),(5,100)\}$ and $n\in\{400,800,1600\}$.
%Table \ref{tab:addSim} shows the mean squared errors of the proposed
%estimator. For each choice of $(J,B)$, the mean squared error decreases
%with the sample size. The mean squared error decreases with $J$,
%and is relatively insensitive to the choice of $B$. The mean squared
%error remains small across all cases.
%
%\begin{table}[t]
%	\caption{\label{tab:addSim}Simulation results for different tuning parameters:
%		mean squared errors ($\times10^{-3}$) of the proposed estimator of
%		$\tau$ for different choices of $(J,B)$ based on $2,000$ Monte
%		Carlo samples}
%	\centering{}%
%	\begin{tabular}{lccc}
%		$(J,B)$  & \multicolumn{1}{c}{$n=400$} & \multicolumn{1}{c}{$n=800$} & \multicolumn{1}{c}{$n=1600$}\tabularnewline
%		$(3,50)$  & $26.8$  & $13.9$  & $8.3$\tabularnewline
%		$(3,100)$  & $27.0$  & $14.1$  & $8.7$\tabularnewline
%		$(5,50)$  & $19.5$  & $9.7$  & $4.1$\tabularnewline
%		$(5,100)$  & $21.3$  & $10.2$  & $4.5$\tabularnewline
%	\end{tabular}
%\end{table}

\section{Asymptotic results for the nonparametric estimation\label{sec:Asymptotic-results}}

We study the consistency of the proposed estimator $\hat{\tau}$ of
$\tau$. The literature has established comprehensive consistency
results for nonparametric estimators and series estimators. For completeness
of our theory, in $\mathsection$\ref{sec1} and $\mathsection$\ref{sec2},
we establish the consistency of the nonparametric estimators in Step
\ref{step1} and the series estimator of $\xi_{ra}(X)$ in Step \ref{step2},
which serve building blocks for deriving the consistency result for
$\hat{\tau}$ in $\mathsection$\ref{sec3}.

\subsection{The consistency of the nonparametric estimators in Step \ref{step1}\label{sec1}}

We assume that the kernel functions and the bandwidth $h_{n}$ satisfy
the following regularity conditions:

\begin{assumption}\label{appendC} \textcolor{black}{(i)$\int_{\R^{p}}K(s)\de s=1$;
		(ii) $||K||_{\infty}=\sup_{x\in\R^{p}}|K(x)|=\kappa<\infty$; (iii)
		$K(\cdot)$ is right continuous; (iv) $\int_{\R^{p}}\Psi_{K}(x)\de x<\infty,$
		where $\Psi_{K}(x)=\sup_{||y||\geq||x||}|K(y)|$, for $x\in\R^{p}$;
		and (v) the kernel function is regular and satisfies the following
		uniform entropy condition. Let $\mathcal{K}$ be the class of functions
		indexed by $x$, 
		\[
		\mathcal{K}=\left\{ K\left(\frac{x-\cdot}{h^{1/p}}\right):h>0,x\in\R^{p}\right\} .
		\]
		Suppose $\mathcal{B}$ is a Borel set in $\R^{p}$, and $Q$ is some
		probability measure on $(\R^{p},\mathcal{B})$. Define $d_{Q}$ to
		be the $L_{2}(Q)$-metric, and $N(\epsilon,\mathcal{K},d_{Q})$ the
		minimal number of balls $\{g:d_{Q}(g,g')<\epsilon\}$ of $d_{Q}$-radius
		$\epsilon$ needed to cover $\mathcal{K}.$ Let $N(\epsilon,\mathcal{K})=\sup_{Q}N(\kappa\epsilon,\mathcal{K},d_{Q})$,
		where the supremum is taken over all probability measures $Q$. For
		some $C>0$ and $\nu>0$, $N(\epsilon,\mathcal{K})\leq C\epsilon^{-\nu}$
		for any $0<\epsilon<1$.}
	
\end{assumption}

\begin{assumption}\label{cond_bw} \textcolor{black}{$h_{n}$ decreases
		to zero, $h_{n}/h_{2n}$ is bounded, $\log(1/h_{n})/\log\log n\rightarrow\infty$
		and $nh_{n}/\log n\rightarrow\infty$, as $n\rightarrow\infty$. }
	
\end{assumption}

\textcolor{black}{\citet{van1996weak} provides sufficient conditions
	for Assumption \ref{appendC} (v).}

\begin{lemma}[Consistency of kernel density estimators]\label{Thm:3 consistency of kernel-based}Let
	$\hat{f}(X\mid A=a,R=1_{p})$ be the kernel density estimator of $f(X\mid A=a,R=1_{p})$,
	where the kernel function satisfies Assumption \ref{appendC}, and
	the bandwidth $h_{n}$ satisfies Assumption \ref{cond_bw}. Suppose
	that the true density function $f(X\mid A=a,R=1_{p})$ is bounded
	and uniformly continuous in $X$, then 
	\begin{equation}
	\lim_{n\rightarrow\infty}\left\Vert \hat{f}(X\mid A=a,R=1_{p})-f(X\mid A=a,R=1_{p})\right\Vert _{\infty}=0\label{eq:thm1}
	\end{equation}
	almost surely.
	
\end{lemma}

The Nadaraya\textendash Watson estimators of $E(Y\mid A=1,X,R=1_{p})$
and $E(Y\mid A=0,X,R=1_{p})$ are 
\begin{eqnarray}
\hat{E}(Y\mid A=1,X,R=1_{p})
&=&\frac{\sum_{i:R_{i}=1_{p}}A_{i}Y_{i}K\left(\frac{X-X_{i}}{h_{n}^{1/p}}\right)}{\sum_{i:R_{i}=1_{p}}A_{i}K\left(\frac{X-X_{i}}{h_{n}^{1/p}}\right)},\label{NW1}
\\ 
\hat{E}(Y\mid A=0,X,R=1_{p})
&=&\frac{\sum_{i:R_{i}=1_{p}}(1-A_{i})Y_{i}K\left(\frac{X-X_{i}}{h_{n}^{1/p}}\right)}{\sum_{i:R_{i}=1_{p}}(1-A_{i})K\left(\frac{X-X_{i}}{h_{n}^{1/p}}\right)},\label{NW2}
\end{eqnarray}
respectively. In this article, we focus on the Nadaraya\textendash Watson
estimator, but we can also consider other nonparametric estimators,
such as local polynomial estimator.

\textcolor{black}{Let $I$ be a compact subset of $\R^{p}$. For any
	function $\psi:\R^{p}\rightarrow\mathcal{\R}$, define 
	\begin{equation}
	||\psi||_{I}=\sup_{x\in I}|\psi(X)|.\label{eq:define norm}
	\end{equation}
	Also, denote $I^{\epsilon}=\{X\in\mathcal{\R}^{p}:\max_{1\leq i\leq p}|X_{i}|\leq\epsilon\}$. }

\begin{lemma}[Consistency of kernel-based estimators for conditional means]\label{Thm: 5 consistency of kernel-based-1}Suppose
	that the kernel function $K(\cdot)$ in (\ref{NW1}) and (\ref{NW2})
	satisfies Assumption \ref{appendC} with support contained in $[-1/2,1/2]^{p}$,
	and the bandwidth $h_{n}$ satisfies Assumption \ref{cond_bw}.
	
	Suppose that there exists an $\epsilon>0$ such that $f(X\mid A=a,R=1_{p})=\int_{-\infty}^{\infty}f(X,Y\mid A=a,R=1_{p})\de Y$
	is continuous and strictly positive on $I^{\epsilon}$, and that $f(X,Y\mid A=a,R=1_{p})$
	is continuous in $X$ for almost every $Y\in\R$. Suppose further
	that there exists an $M>0$ such that for $X\in I^{\epsilon}$, $|Y|\leq M$
	almost surely. Then, for any $a$, 
	\begin{equation}
	\lim_{n\rightarrow\infty}\left\Vert \hat{E}(Y\mid A=a,X,R=1_{p})-E(Y\mid A=a,X,R=1_{p})\right\Vert _{I}=0\label{eq:thm2}
	\end{equation}
	almost surely.
	
\end{lemma}

A large literature has developed consistency of kernel-based estimators.
The proofs of Lemmas \ref{Thm:3 consistency of kernel-based} and
\ref{Thm: 5 consistency of kernel-based-1} are similar to those given
by \citet{deheuvels2000uniform} and \citet{gine2002rates}, and therefore
are omitted. The smoothing spline estimator is asymptotically equivalent
to a kernel-based estimator that employs the so-called spline kernel
\citep{silverman1984spline}. Both spline kernels and Gaussian kernels
satisfy Assumption \ref{appendC} \citep{van1996weak}. Therefore,
by Lemmas \ref{Thm:3 consistency of kernel-based} and \ref{Thm: 5 consistency of kernel-based-1},
the nonparametric estimators in Step \ref{step1} are consistent.

\subsection{The consistency of the series estimator of $\xi_{ra}(X)$ in Step
	\ref{step2}\label{sec2}}

For any $r$ and $a$, $\xi_{ra}(X)$ satisfies the conditional moment
restriction\textcolor{black}{  
	\[
	E\left\{ f(X_{\obs},Y,R=r\mid A=a)-\xi_{ra}(X)f(X,Y,R=1_{p}\mid A=a)\mid A=a,X_{\obs},Y,R=r\right\} =0.
	\]
	We define a generalized residuals with the function of interest $h(X)$
	as 
	\[
	\rho_{ra}(X,Y;h)=f(X_{\obs},Y,R=r\mid A=a)-h(X)f(X,Y,R=1_{p}\mid A=a),
	\]
	the conditional mean function of $\rho_{ra}(X,Y;h)$ given $(A=a,X_{\obs},Y,R=r)$
	as 
	\[
	m_{ra}(X_{\obs},Y;h)=E\{\rho_{ra}(X,Y;h)\mid A=a,X_{\obs},Y,R=r\},
	\]
	and the series least square estimator of the conditional mean function
	as 
	\begin{eqnarray*}
		\hat{m}_{ra}(X_{\obs},Y;h) & = & \hat{f}(X_{\obs},Y,R=r\mid A=a)\\
		&  & -\hat{E}\{h(X)\mid A=a,X_{\obs},Y,R=r\}\hat{f}(X,Y,R=1_{p}\mid A=a).
	\end{eqnarray*}
	Following these definitions, $m_{ra}(X_{\obs},Y;\xi_{ra})=0$ for
	any $r$ and $a$. Let the project of $\xi_{ra}$ onto $\mathcal{H}_{J}$
	be $\prod_{\mathcal{H}_{J}}\xi_{ra}(\cdot)=\sum_{j=1}^{J}\beta_{ra}^{j}h^{j}(\cdot)$
	such that $||\prod_{\mathcal{H}_{J}}\xi_{ra}-\xi_{ra}||_{\infty}=o(1)$. }

To avoid technicality, we assume the following high-level regularity
conditions.

\begin{assumption}\label{appendixB} (i) $E\{||m_{ra}(X_{\obs},Y;\prod_{\mathcal{H}_{J}}\xi_{ra})||_{\mathcal{G}}^{2}\}=o(1)$;
	(ii) $n^{-1}\sum_{i=1}^{n}||m_{ra}(X_{\obs,i},Y_{i};\prod_{\mathcal{H}_{J}}\xi_{ra})||_{\mathcal{G}}^{2}$
	$\leq c_{0}E||m_{ra}(X_{\obs},Y;\prod_{\mathcal{H}_{J}}\xi_{ra})||_{\mathcal{G}}^{2}+o_{p}(1)$
	and a finite constant $c_{0}>0$; (iii) $n^{-1}\sum_{i=1}^{n}||\hat{m}_{ra}(X_{\obs,i},Y_{i};h)||_{\mathcal{G}}^{2}\geq c_{1}E||m_{ra}(X_{\obs},Y;h)||_{\mathcal{G}}^{2}-o_{p}(1)$
	uniformly for $h$ over $\mathcal{H}_{J}$ and a finite constant $c_{1}>0$.
	
\end{assumption}

Assumption \ref{appendixB} (i) holds if $E\{||m_{ra}(X_{\obs},Y;h)||_{\mathcal{G}}^{2}\}$
is continuous at $h=\xi_{ra}$ under $||\cdot||_{\infty}$. Assumption
\ref{appendixB} (ii) and (iii) are sample criteria to regularize
the asymptotic behavior of the series estimator of $m_{ra}(X_{\obs,i},Y_{i};h)$.
\citet{chen2012estimation} provided sufficient conditions for Assumption
\ref{appendixB}.

\begin{lemma}[Consistency of $\hat{\xi}_{ra}$]\label{Thm:4}\textcolor{black}{Under
		Assumptions \ref{asump-ignorable}, \ref{asump:Pos}, \ref{asump:OIMissing}
		and Assumption \ref{appendixB}, the series estimator $\hat{\xi}_{ra}(X)=\sum_{j=1}^{J}\hat{\beta}_{ra}^{j}h^{j}(\tilde{X})$
		is consistent for $\xi_{ra}(X)$ in the sense that $||\hat{\xi}_{ra}-\xi_{ra}||_{\infty}=o_{p}(1)$
		as $J\rightarrow\infty$ and $n\rightarrow\infty$. }
	
\end{lemma}

\citet{chen2012estimation} provided a proof for Lemma \ref{Thm:4}
in the context of estimation of nonparametric conditional moment models.
Our proof for Lemma \ref{Thm:4} is similar, and therefore omitted.

\subsection{The consistency of the proposal estimator of $\tau$ in Step \ref{step4}\label{sec3}}

Let $||X||$ be the Euclidean norm for $X$. Denote $I_{K}=\{X:||X||>K\}$
for a constant $K$, and $I_{K}^{c}$ to be the complement set of
$I_{K}$.

\begin{theorem}[Consistency of $\hat{\tau}$]\label{Thm:main} Suppose
	that the assumptions in Theorem \ref{Thm:tau} and Lemmas \ref{Thm:3 consistency of kernel-based}\textendash \ref{Thm:4}
	hold. Suppose further that for some $B>0$, $\hat{\tau}(X)$ and $\tau(X)$
	are uniformly bounded for $X\in I_{B}$, and that 
	\begin{equation}
	\int_{I_{K}^{c}}\frac{f(X\mid A=a,R=1_{p})}{\pr(R=1_{p}\mid A=a,X)}\de\nu(X)\rightarrow0,\label{eq:tail prop}
	\end{equation}
	as $K\rightarrow\infty$. Then, the nonparametric estimator $\hat{\tau}$
	resulting from (\ref{eq:tauhat}) is consistent for $\tau$.
	
\end{theorem}

The proposed estimator $\hat{\tau}$ is a linear functional of $\hat{\tau}(\cdot)$,
$\hat{f}(\cdot\mid A=a,R=1_{p})$, and $\hat{\xi}_{ra}(\cdot)$. A
large literature has established the root-$n$ asymptotic normality
and the consistent variance estimation for plug-in series estimators
of functionals; see, for example, \citet{newey1997convergence}, \citet{shen1997methods},
\citet{chen1998sieve}, \citet{li2007nonparametric}, \citet{chen2007large},
\citet{chen2009efficient}, \citet{chen2012estimation}, and \citet{chen2014sieve}.
Alternatively, \citet{chen2015sieve}\textcolor{black}{{} provided
	Wald and quasi-likelihood ratio inference results for the general
	models in \citet{chen2012estimation}, including series two stage
	least squares as an example.} A relatively simple approach is to treat
the nonparametric estimators as if they were parametric given the
fixed tuning parameters, so that there is only a finite number of
parameters. From this point of view, we can use standard approaches
for variance estimation under parametric models. This approach is
asymptotically valid for nonparametric series regression; see, for
example, \citet{newey1997convergence}. In the light of treating the
nonparametric estimators as if they were parametric, one might expect
the nonparametric bootstrap to work for our estimator. For all bootstrap
samples, we use the same tuning parameters, such as the smoothing
parameter in the smoothing splines and the bandwidth in the kernel
density estimator. In our simulation study, inference based on the
above bootstrap is promising. However, it is a difficult task
(if it is possible) to prove that the bootstrap is consistent which
is beyond the scope of this article. Recent work has shown
that it does work for some nonparametric instrumental variable series
estimators \citep{horowitz2007asymptotic}. 

\subsection{Proof of Theorem \ref{Thm:main}}

By Lemmas \ref{Thm:3 consistency of kernel-based} and \ref{Thm: 5 consistency of kernel-based-1},
\begin{equation}
\lim_{n\rightarrow\infty}\left|\left|\frac{\hat{f}(X\mid A=a,R=1_{p})}{\widehat{\pr}(R=1_{p}\mid A=a,X)}-\frac{f(X\mid A=a,R=1_{p})}{\pr(R=1_{p}\mid A=a,X)}\right|\right|_{\infty}=0\label{eq:p1}
\end{equation}
almost surely. Since $\hat{\tau}(X)$ and $\tau(X)$ are uniformly
bounded in $I_{K}$ for $K>B$, together with (\ref{eq:tail prop})
and (\ref{eq:p1}), for any $\epsilon$, there exists $K_{2}>0$,
such that for any $K>K_{2}$,{\small{}{} 
	\begin{multline}
	\lim_{n\rightarrow\infty}\pr\left[\left|\int_{I_{K}^{c}}\hat{\tau}(X)\left\{ \frac{\hat{f}(X\mid A=a,R=1_{p})}{\widehat{\pr}(R=1_{p}\mid A=a,X)}-\frac{f(X\mid A=a,R=1_{p})}{\pr(R=1_{p}\mid A=a,X)}\right\} \de\nu(X)\right|>\frac{\epsilon}{4}\right]<\frac{\epsilon}{4},\label{eq:pp2}
	\end{multline}
}and 
\begin{equation}
\lim_{n\rightarrow\infty}\left|\int_{I_{K}^{c}}\left\{ \hat{\tau}(X)-\tau(X)\right\} \frac{f(X\mid A=a,R=1_{p})}{\pr(R=1_{p}\mid A=a,X)}\de\nu(X)\right|<\frac{\epsilon}{4}.\label{eq:pp3}
\end{equation}

By Theorem \ref{Thm:4}, for any $K$,{\small{}{} 
	\begin{equation}
	\lim_{n\rightarrow\infty}\left|\left|\hat{\tau}(X,R=1_{p})\frac{\hat{f}(X\mid A=a,R=1_{p})}{\widehat{\pr}(R=1_{p}\mid A=a,X)}-\tau(X,R=1_{p})\frac{f(X\mid A=a,R=1_{p})}{\pr(R=1_{p}\mid A=a,X)}\right|\right|_{I_{K}}=0\label{eq:p2}
	\end{equation}
}almost surely, where $||\cdot||_{I}$ is defined in (\ref{eq:define norm}).
Therefore, for any $\epsilon$, by (\ref{eq:p2}), we choose $K_{1}$
such that for any $K>K_{1}$, 
\begin{multline}
\lim_{n\rightarrow\infty}\pr\left\{ \left|\int_{I_{K_{1}}}\hat{\tau}(X,R=1_{p})\frac{\hat{f}(X\mid A=a,R=1_{p})}{\widehat{\pr}(R=1_{p}\mid A=a,X)}\de\nu(X)\right.\right.\\
\left.\left.\quad\quad\quad\quad\quad-\int_{I_{K_{1}}}\tau(X,R=1_{p})\frac{f(X\mid A=a,R=1_{p})}{\pr(R=1_{p}\mid A=a,X)}\de\nu(X)\right|>\frac{\epsilon}{2}\right\} <\frac{\epsilon}{2}.\label{eq:pp1}
\end{multline}

Combing (\ref{eq:pp2}), (\ref{eq:pp3}) and (\ref{eq:pp1}), for
any $\epsilon>0$, we choose $K>\max(K_{1},K_{2})$, 
\begin{eqnarray*}
	&  & \lim_{n\rightarrow\infty}\pr(|\hat{\tau}-\tau|>\epsilon)\\
	& = & \lim_{n\rightarrow\infty}\pr\left\{ \left|\int\hat{\tau}(X)\frac{\hat{f}(X\mid A=a,R=1_{p})}{\widehat{\pr}(R=1_{p}\mid A=a,X)}\de\nu(X)-\int\tau(X)\frac{f(X\mid A=a,R=1_{p})}{\pr(R=1_{p}\mid A=a,X)}\de\nu(X)\right|>\epsilon\right\} \\
	& \leq & \lim_{n\rightarrow\infty}\pr\left\{ \left|\int_{I_{K}}\hat{\tau}(X)\frac{\hat{f}(X\mid A=a,R=1_{p})}{\widehat{\pr}(R=1_{p}\mid A=a,X)}\de\nu(X)-\int_{I_{K}}\tau(X)\frac{f(X\mid A=a,R=1_{p})}{\pr(R=1_{p}\mid A=a,X)}\de\nu(X)\right|>\frac{\epsilon}{2}\right\} \\
	& + & \lim_{n\rightarrow\infty}\pr\left[\left|\int_{I_{K}^{c}}\hat{\tau}(X)\left\{ \frac{\hat{f}(X\mid A=a,R=1_{p})}{\widehat{\pr}(R=1_{p}\mid A=a,X)}-\frac{f(X\mid A=a,R=1_{p})}{\pr(R=1_{p}\mid A=a,X)}\right\} \de\nu(X)\right|>\frac{\epsilon}{4}\right]\\
	& + & \lim_{n\rightarrow\infty}\pr\left[\left|\int_{I_{K}^{c}}\left\{ \hat{\tau}(X)-\tau(X)\right\} \frac{f(X\mid A=a,R=1_{p})}{\pr(R=1_{p}\mid A=a,X)}\de\nu(X)\right|>\frac{\epsilon}{4}\right]<\epsilon,
\end{eqnarray*}
that is, $\hat{\tau}$ is consistent for $\tau$.

\section{More details for the parametric estimation of $\tau$}

\subsection{An example of a bounded complete distribution}

The bounded completeness is a weaker concept than the completeness.
We say that a function $f(X,Y)$ is complete in $Y$ if $\int g(X)f(X,Y)\de\nu(X)=0$
implies $g(X)=0$ almost surely for any squared integrable function
$g(X)$. For illustration, we give sufficient conditions for the completeness
of distribution functions in an exponential family, which implies
the bounded completeness.

\begin{lemma}\label{lemma1} The distribution $f(X,Y)=\psi(X)h(Y)\exp\{\lambda(Y)^{\T}\eta(X)\}$
	is bounded complete in $y$ if (i) $\psi(X)>0$, (ii) $\lambda(Y)>0$
	for $Y\in\mathcal{B}$ when $\mathcal{B}$ is an open set, and (iii)
	the mapping $X\mapsto\eta(X)$ is one-to-one. \end{lemma}

\begin{proof} Suppose that $\int g(X)f(X,Y)\de\nu(X)=0$,
	which, in this setting, is 
	\begin{equation}
	h(Y)\int\tilde{g}(X)\exp\{\lambda(Y)^{\T}\eta(X)\}\de\nu(X)=0,\label{eq:integral}
	\end{equation}
	where $\tilde{g}(X)=g(X)\psi(X)$. Since the mapping $X\mapsto\eta(X)$
	is one-to-one, let $T=\eta(X)$ and therefore $X=\eta^{-1}(T)$. Then,
	the integral equation (\ref{eq:integral}) becomes 
	\begin{equation}
	h(Y)\int\tilde{g}\{\eta^{-1}(T)\}[\dot{\eta}\{\eta^{-1}(T)\}]^{-1}\exp\{\lambda(Y)^{\T}T\}\de\nu(T)=0,\label{eq:integral2}
	\end{equation}
	and particularly for $Y\in\mathcal{B}$, where $[\dot{\eta}\{\eta^{-1}(T)\}]^{-1}$
	is the Jacobian matrix with $\dot{\eta}(x)=\partial\eta(x)/\partial x$.
	The left hand side of the integral equation (\ref{eq:integral2})
	as a function of $\lambda(Y)$ is a multivariate Laplace transform
	of $\tilde{g}\{\eta^{-1}(T)\}[\dot{\eta}\{\eta^{-1}(T)\}]^{-1}$,
	and it cannot be zero unless $\tilde{g}\{\eta^{-1}(T)\}[\dot{\eta}\{\eta^{-1}(T)\}]^{-1}$
	is zero almost everywhere. Since $[\dot{\eta}\{\eta^{-1}(T)\}]^{-1}$
	is not zero, (\ref{eq:integral2}) holds only if $\tilde{g}(X)$ is
	zero almost everywhere. Moreover, since $\psi(X)$ is not zero, $g(X)$
	is zero almost everywhere. This completes the proof. \end{proof}

\begin{proposition} \label{eg:Gaussian}The Gaussian model 
	\begin{equation}
	f(X,Y)=f(Y\mid X)f(X)=\frac{1}{(2\pi\sigma^{2})^{1/2}}\exp\left\{ -\frac{(Y-\beta_{0}-\beta_{1}^{\T}X)^{2}}{2\sigma^{2}}\right\} f(X),\label{eq:2}
	\end{equation}
	is bounded complete in $Y$, where $\beta_{1}=(\beta_{11},\ldots,\beta_{1p})^{\T}$
	and $X=(X_{1},\ldots,X_{p})^{\T}$.
	
\end{proposition}

\begin{proof}
	Using the notation in Lemma \ref{lemma1}, (\ref{eq:2}) can be expressed
	as $f(X,Y)=\psi(X)\exp\{\lambda(Y)^{\T}\eta(X)\}$ with $\psi(X)=(2\pi\sigma^{2})^{-1/2}f(X)$,
	$\lambda(Y)=\sigma^{-2}(\beta_{11}Y,\ldots,\beta_{1p}Y)^{\T}$ and
	$\eta(X)=(X_{1},\ldots,X_{p})^{\T}$. Therefore, (\ref{eq:2}) satisfies
	the conditions for $\lambda(Y)$ and $\eta(X)$, and it is bounded
	complete in $Y$. \end{proof}

\subsection{Likelihood-based inference: a fractional imputation approach\label{subsec:Frequentist-perspective}}

Let $S(\theta;Z_{i})=\partial\log f(Z_{i};\theta)/\partial\theta$
be the complete-data score for unit $i.$ The maximum likelihood estimator
$\hat{\theta}$ is a solution of the conditional score equation  \citep{kim2013statistical}
\begin{equation}
n^{-1}\sum_{i=1}^{n} \sum_{r\in\mathcal{R}} I(R_i=r)E\{S(\theta;Z_{i})\mid Z_{\obs,i},R_i=r;\theta\}=0,\label{eq:mean-score}
\end{equation}
where the conditional expectation is with respect to 
\begin{equation}
f(X_{\mis,i}\mid Z_{\obs,i},R_{i}=r;\theta)=\frac{f(A_{i},X_{i},Y_{i},R_{i}=r;\theta)}{\int f(A_{i},X_{i},Y_{i},R_{i}=r;\theta)\de\nu(X_{\mis,i})}.\label{eq:cond dist}
\end{equation}
The EM algorithm is a standard tool for solving \eqref{eq:mean-score}.
However, it has several drawbacks. First, the computation of the conditional
expectation in (\ref{eq:mean-score}) can be difficult due to the
possibly high-dimensional integration. Second, the conditional distribution
(\ref{eq:cond dist}) may not have an explicit form. We can use the
fractional imputation \citep{yang2015fractional} to overcome the
computation difficulties. The fractional imputation uses importance
sampling to avoid analytical calculation for evaluating the conditional
expectation.

In fractional imputation, we approximate the conditional expectation
in \eqref{eq:mean-score} by 
\begin{equation}
\sum_{r\in\mathcal{R}} I(R_i=r) E\{\tau(Z_{i};\theta)\mid Z_{\obs,i},R_i=r;\theta\}=\sum_{j=1}^{M}\omega_{ij}^{*}\tau(Z_{ij}^{*};\theta)+o_{p}(M^{-1/2}),\label{1-3}
\end{equation}
where $\{Z_{ij}^{*}=(A_{i},X_{\obsR,i},X_{\misR,i}^{*(j)},Y_{i},R_{i}):j=1,\ldots,M\}$
are the fractional observations and the $\omega_{ij}^{*}$'s are the
fractional weights that satisfies $\omega_{ij}^{*}\ge0$ and $\sum_{j=1}^{M}\omega_{ij}^{*}=1$.
Approximately, we can solve $\hat{\theta}$ from 
\begin{eqnarray}
\frac{1}{n}\sum_{i=1}^{n}\sum_{j=1}^{M}\omega_{ij}^{*}S(\theta;Z_{ij}^{*})=0.\label{eq::eq-fractional}
\end{eqnarray}
Computationally, we iteratively generate weighted fractional observations
satisfying \eqref{1-3} and solve the conditional score equation \eqref{eq::eq-fractional}.
This often converges to $\hat{\theta}$.

The key is to construct \eqref{1-3} using importance sampling. For
each missingness pattern $R_i=r$ and the missing value $X_{\mis,i}$, we first generate $X_{\mis,i}^{*(1)},\ldots,X_{\mis,i}^{*(M)}$
from a proposal distribution $h(X_{\mis,i}\mid Z_{\obs,i})$ for some
$h(\cdot)$ that is easy to simulate. We then compute 
\[
\omega_{ij}^{*}\propto\frac{f(X_{\mis,i}^{*(j)}\mid Z_{\obs,i};\hat{\theta})}{h(X_{\mis,i}^{*(j)}\mid Z_{\obs,i})}\propto\frac{f(Z_{ij}^{*};\hat{\theta})}{h(X_{\mis,i}^{*(j)}\mid Z_{\obs,i})},
\]
subject to $\sum_{j=1}^{M}\omega_{ij}^{*}=1$, as the fractional weight
for $Z_{ij}^{*}$.

As a by product, we can also use 
\[
\tilde{\tau}(\hat{\theta})=n^{-1}\sum_{i=1}^{n}\sum_{j=1}^{M}\hat{\omega}_{ij}\tau(\hat{Z}_{ij};\hat{\theta})
\]
as an estimator for $\tau$, where the $\hat{\omega}_{ij}$'s are
the weights for the fractional observations $\hat{Z}_{ij}$'s at the
maximum likelihood estimator $\hat{\theta}$. Clearly, $\tilde{\tau}(\hat{\theta})$
is an approximation to 
\[
\tilde{\tau}(\theta)=n^{-1}\sum_{i=1}^{n} \sum_{r\in\mathcal{R}} I(R_i=r)E\{\tau(X_{i};\theta)\mid Z_{\obs,i},R_i=r;\theta\},
\]
which satisfies $E\{\tilde{\tau}(\theta)\}=\tau.$

\subsection{Bayesian approach: an example with a scalar $X$}

Let $R$ be the missing indicator the scalar $X$. Suppose 
\begin{eqnarray*}
	\pr(R=0\mid A=a,X,Y;\eta) & = & \pr(R=0\mid A=a,X;\eta)=\left\{ 1+\exp(\eta_{a0}+\eta_{a1}X)\right\} ^{-1},\\
	f(Y\mid A=a,X;\beta) & = & (2\pi\sigma_{a}^{2})^{-1/2}\exp\{-(Y-\beta_{a0}-\beta_{a1}X)^{2}/(2\sigma_{a}^{2})\},\\
	\pr(A=1\mid X;\alpha) & = & \text{logit}(\alpha_{0}+\alpha_{1}X),\\
	f(X;\lambda) & = & (2\pi\sigma_{x}^{2})^{-1/2}\exp\{-(X-\mu_{x})^{2}/(2\sigma_{x}^{2})\},
\end{eqnarray*}
where $\eta=(\eta_{00},\eta_{01},\eta_{10},\eta_{11})$, $\beta=(\beta_{00},\beta_{01},\sigma_{0}^{2},\beta_{10},\beta_{11},\sigma_{1}^{2})$,
$\alpha=(\alpha_{0},\alpha_{1})$, and $\lambda=(\mu_{x},\sigma_{x}^{2})$.
The parametric $\theta=(\alpha,\beta,\eta,\lambda)$ has prior $\pi(\theta)$.
The complete-data likelihood is $L(\theta\mid Z_{1},\ldots,Z_{n})=\prod_{i=1}^{n}f(Z_{i};\theta),$
where 
\begin{eqnarray}
f(Z_{i};\theta) & = & \left[\frac{\exp(\eta_{10}+\eta_{11}X_{i})^{R_{i}}}{1+\exp(\eta_{10}+\eta_{11}X_{i})}\frac{1}{\left(2\pi\sigma_{1}^{2}\right)^{1/2}}\exp\left\{ -\frac{(Y_{i}-\beta_{10}-\beta_{11}X_{i})^{2}}{2\sigma_{1}^{2}}\right\} \right]^{A_{i}}\nonumber \\
&  & \times\left[\frac{\exp(\eta_{00}+\eta_{01}X_{i})^{R_{i}}}{1+\exp(\eta_{00}+\eta_{01}X_{i})}\frac{1}{\left(2\pi\sigma_{0}^{2}\right)^{1/2}}\exp\left\{ -\frac{(Y_{i}-\beta_{00}-\beta_{01}X_{i})^{2}}{2\sigma_{0}^{2}}\right\} \right]^{1-A_{i}}\nonumber \\
&  & \times\frac{\exp(\alpha_{0}+\alpha_{1}X_{i})^{A_{i}}}{1+\exp(\alpha_{0}+\alpha_{1}X_{i})}\times\frac{1}{\left(2\pi\sigma_{x}^{2}\right)^{1/2}}\exp\left\{ -\frac{(X_{i}-\mu_{x})^{2}}{2\sigma_{x}^{2}}\right\} .\label{eq:f(zi)}
\end{eqnarray}
By Lemma \ref{lemma1}, it is easy to verify that $f(A=a,X,Y,R=1)$
is bounded complete in $Y$. By Theorem \ref{thm2}, $\theta$ is
identifiable.

In the Bayesian estimation, we first simulate the posterior distribution
of the $Z_{i}$'s and $\theta$. Given the parameter value $\theta^{*}=(\alpha^{*},\beta^{*},\eta^{*},\lambda^{*}),$
we generate 
\begin{eqnarray*}
	X_{i}^{*} & \sim & f(X_{i}\mid A_{i},Y_{i},R_{i};\theta^{*})\\
	& \propto & \left[\frac{\exp(\eta_{10}^{*}+\eta_{11}^{*}X_{i})^{R_{i}}}{1+\exp(\eta_{10}^{*}+\eta_{11}^{*}X_{i})}\frac{1}{\left(2\pi\sigma_{1}^{*2}\right)^{1/2}}\exp\left\{ -\frac{(Y_{i}-\beta_{10}^{*}-\beta_{11}^{*}X_{i})^{2}}{2\sigma_{1}^{*2}}\right\} \right]^{A_{i}}\\
	&  & \times\left[\frac{\exp(\eta_{00}^{*}+\eta_{01}^{*}X_{i})^{R_{i}}}{1+\exp(\eta_{00}^{*}+\eta_{01}^{*}X_{i})}\frac{1}{\left(2\pi\sigma_{0}^{*2}\right)^{1/2}}\exp\left\{ -\frac{(Y_{i}-\beta_{00}^{*}-\beta_{01}^{*}X_{i})^{2}}{2\sigma_{0}^{*2}}\right\} \right]^{1-A_{i}}\\
	&  & \times\frac{\exp(\alpha_{0}^{*}+\alpha_{1}^{*}X_{i})^{A_{i}}}{1+\exp(\alpha_{0}^{*}+\alpha_{1}^{*}X_{i})}\frac{1}{\left(2\pi\sigma_{x}^{*2}\right)^{1/2}}\exp\left\{ -\frac{(X_{i}-\mu_{x}^{*})^{2}}{2\sigma_{x}^{*2}}\right\} 
\end{eqnarray*}
for units with $R_{i}=0$. For units with $R_{i}=1$, let $X_{i}^{*}=X_{i}.$
Given the imputed values $X_{i}^{*}$, we have the complete data $Z_{i}^{*}$,
and then generate $\theta^{*}\sim f(\theta\mid Z_{1}^{*},\ldots,Z_{n}^{*})\propto L(\theta\mid Z_{1}^{*},\ldots,Z_{n}^{*})\pi(\theta).$
Both steps may involve the Markov chain Monte Carlo.

Given $(\theta^{*},X_{1}^{*},\ldots,X_{n}^{*})$, we calculate $\hat{\tau}(\theta^{*})=n^{-1}\sum_{i=1}^{n}\tau(X_{i}^{*};\theta^{*})$
as a posterior draw of $\hat{\tau}(\theta)$. This gives the posterior
distribution of the average causal effect conditioning on the covariate
values.

\end{document}